\documentclass[]{interact}
\usepackage{epstopdf}
\usepackage[caption=false]{subfig}
\usepackage{caption}
\usepackage[numbers,sort&compress]{natbib}
\bibpunct[, ]{[}{]}{,}{n}{,}{,}
\renewcommand\bibfont{\fontsize{10}{12}\selectfont}

\theoremstyle{plain}
\newtheorem{theorem}{Theorem}[section]
\newtheorem{lemma}[theorem]{Lemma}

\theoremstyle{definition}
\newtheorem{definition}[theorem]{Definition}

\theoremstyle{remark}
\newtheorem{remark}{Remark}

\newcommand{\Prob}{\mathsf{P}}
\newcommand{\Expect}{\mathsf{E}}
\DeclareMathOperator*{\esssup}{ess\,sup}

\begin{document}


\title{Quickest Change Detection in Statistically Periodic Processes with Unknown Post-Change Distribution}

\author{
\name{Yousef Oleyaeimotlagh\textsuperscript{a}, Taposh Banerjee\textsuperscript{a}\thanks{CONTACT Taposh Banerjee. Email: taposh.banerjee@pitt.edu}, Ahmad Taha\textsuperscript{b}, and Eugene John\textsuperscript{c}}
\affil{\textsuperscript{a}University of Pittsburgh, \textsuperscript{b}Vanderbilt University, \textsuperscript{c}University of Texas at San Antonio
}
}
\maketitle

\begin{abstract}
Algorithms are developed for the quickest detection of a change in statistically periodic processes. These are processes in which the statistical properties are nonstationary but repeat after a fixed time interval. It is assumed that the pre-change law is known to the decision maker but the post-change law is unknown. In this framework, three families of problems are studied: robust quickest change detection, joint quickest change detection and classification, and multislot quickest change detection. In the multislot problem, the exact slot within a period where a change may occur is unknown. Algorithms are  proposed for each problem, and either exact optimality or asymptotic optimal in the low false alarm regime is proved for each of them. The developed algorithms are then used for anomaly detection in traffic data and arrhythmia detection and identification in electrocardiogram (ECG) data. The effectiveness of the algorithms is also demonstrated on simulated data.  
\end{abstract}


\section{Introduction}

In the classical problem of quickest change detection (see  \cite{poor-hadj-qcd-book-2009}, \cite{tart-niki-bass-2014}, \cite{veer-bane-elsevierbook-2013}),
a decision maker observes a stochastic process with a given distribution. At some point in time, the distribution of the process changes. The
problem objective is to detect this change in distribution as quickly as possible, with minimum possible delay, subject to a constraint
on the rate of false alarms. This problem has applications in statistical process control (\cite{taga-jqt-1998}), sensor networks (\cite{bane-tsp-2015}), cyber-physical system monitoring (\cite{chen-bane-tps-2016}), regime changes in neural data (\cite{bane-NER-2019}), traffic monitoring (\cite{bane-fusion-2018}), and in general, anomaly detection (\cite{bane-fusion-2018, bane-globalsip-2018}).

In many applications of anomaly detection, the observed process has statistically periodic behavior. Some examples are as follows:
\begin{enumerate}
	\item \textit{Arrhythmia detection in ECG Data}: The electrocardiography (ECG) data has an almost periodic waveform pattern with a series of P waves, QRS complexes, and ST segments. An arrhythmia can cause a change in this regular pattern (\cite{PhysioNet}).
	\item \textit{Detecting changes in neural spike data}: In certain brain-computer interface (BCI) studies (\cite{zhang-demba-2018}), an identical experiment is performed on an animal in a series of trials leading to similar firing patterns in each trial. An event or a trigger (which is part of the experiment) can change the firing pattern after a certain trial.
	\item \textit{Anomaly detection in city traffic data}: The count of vehicles at a street intersection in New York City (NYC) has been found to show regular patterns of busy and quiet periods (\cite{bane-fusion-2018, bane-globalsip-2018, bane-icassp-2019, bane-isit-2019, bane-allerton-2019}). Congestion or an accident can cause a drop or increase in these vehicle counts.
	\item \textit{Social network data}: The count of Instagram messages posted near a CCTV camera in NYC has also been found to show approximately periodic behavior (\cite{bane-fusion-2018, bane-globalsip-2018, bane-icassp-2019, bane-isit-2019, bane-allerton-2019}).
	\item \textit{Congestion mode detection on highways}: In traffic density estimation problems, it is of interest to detect the mode (congested or uncongested) of the traffic first before deciding on a model to be used for estimation (\cite{taha-bane-acc-2020}). Motivated by the NYC data behavior, the traffic intensity in this application can also be modeled as statistically periodic.
\end{enumerate}

In \cite{bane-tit-2021}, a new class of stochastic processes, called independent and periodically identically distributed (i.p.i.d.) processes, has been introduced to model statistically periodic data. In this process, the sequence of random variables is independent and the distribution of the variables is periodic with a given period $T$.

Statistically periodic processes can also be modeled using cyclostationary processes (\cite{gardner2006cyclostationarity}). However, modeling using i.p.i.d. processes allow for sample-level detection and the development of strong optimality theory.

In \cite{bane-tit-2021}, a Bayesian theory is developed for quickest change detection in i.p.i.d. processes. It is shown that, similar to the i.i.d. setting, it is optimal to use the Shiryaev statistic, i.e., the \textit{a-posteriori} probability that the change has already occurred given the data, for change detection. However, in the i.p.i.d. setting, a change is declared when the sequence of Shiryaev statistics crosses a sequence of time-varying but periodic thresholds.
It is also shown that a single-threshold test is asymptotically optimal, as the constraint on the probability of a false alarm goes to zero. The proposed algorithm can also be implemented recursively and using finite memory. Thus, the set-up of i.p.i.d. processes gives an example of a non-i.i.d. setting in which exactly optimal algorithm can be implemented efficiently. The results in \cite{bane-tit-2021} is valid when both pre- and post-change distributions are known.

In this paper, we consider the problem of quickest change detection in i.p.i.d. processes when the post-change law is unknown. We consider three different formulations of the problem in minimax and Bayesian settings:
\begin{enumerate}
    \item \textit{Robust quickest change detection}: In Section~\ref{sec:robustQCD}, we first consider the problem of robust quickest change detection in i.p.i.d. processes. In this problem, we assume that the post-change family of distributions is not known but belongs to an uncertainty class. We further assume that the post-change family has a distribution that is least favorable. We then show that the algorithm designed using the least favorable distribution is minimax robust for the Bayesian delay metric.
    \item \textit{Quickest detection and fault identification}: 
    In Section~\ref{sec:QCDFaultIsolation}, we consider the problem in which the post-change distribution is unknown but belongs to a finite class of distributions. For this setup, we solve the problem of joint quickest  change detection and isolation in i.p.i.d. processes. We also apply the developed algorithm to real ECG data to detect heart arrhythmia. 
    \item \textit{Multislot quickest change detection}: In Section~\ref{sec:multislot}, we consider the problem of multislot quickest change detection in i.p.i.d. processes. In this problem, the exact time slots in a given period where the change can occur are unknown. We show that a mixture-based test is asymptotically optimal.
\end{enumerate}
A salient feature of our work is that in addition to developing the optimality theory for the proposed algorithms, we also apply them to real or simulation data to demonstrate their effectiveness. Specifically, in Section~\ref{sec:traffic}, we study anomaly detection in traffic data. In Section~\ref{sec:ECG}, we apply the developed algorithms to arrhythmia detection and isolation in ECG data. In Section~\ref{sec:secRobust}, Section~\ref{sec:secMultit}, and Section~\ref{sec:wavelet}, we also apply our algorithms to simulated data to show their effectiveness.

\section{Robust Quickest Change Detection}
\label{sec:robustQCD}

\subsection{Model and Problem Formulation}\label{sec:ProbForm}
We first define the process that we will use to model statistically periodic random processes in this paper.

\medspace
\medspace
\medspace
\begin{definition}[\hspace{-0.03cm}\cite{bane-tit-2021}]
	A random process $\{X_n\}$
	is called independent and periodically identically distributed (i.p.i.d) if
	\begin{enumerate}
		\item The random variables $\{X_n\}$ are independent.
		\item If $X_n$ has density $f_n$, for $n \geq 1$, then there is a positive integer $T$ such
		that the sequence of densities $\{f_n\}$ is periodic with period $T$:
		$$
		f_{n+T} = f_n, \quad \forall n \geq 1.
		$$
	\end{enumerate}
\end{definition}
\medspace
\medspace
\medspace

The law of an i.p.i.d. process is completely characterized by the finite-dimensional product distribution of $(X_1, \dots, X_T)$ or the set of densities $(f_1, \cdots, f_T)$, and we say that the process is i.p.i.d. with the law $(f_1, \cdots, f_T)$.
The change point problem of interest is the following. In the normal regime, the data is modeled as an i.p.i.d. process with law $(f_1, \cdots, f_T)$. At some point in time, due to an event,
the distribution of the i.p.i.d. process deviates from $(f_1, \cdots, f_T)$.
Specifically, consider another periodic sequence of densities $\{g_n\}$ such that
$$
g_{n+T} = g_{n}, \quad \forall n \geq 1.
$$
It is assumed that at the change point $\nu$, the law of the i.p.i.d. process switches from $(f_1, \cdots, f_T)$ to $(g_1, \cdots, g_T)$:
\begin{equation}\label{eq:changepointmodel}
	X_n \sim
	\begin{cases}
		f_n, &\quad \forall n < \nu, \\
		g_n &\quad \forall n \geq \nu.
	\end{cases}
\end{equation}
The densities $(g_1, \cdots, g_T)$ need not be all different from the set of densities $(f_1, \cdots, f_T)$,
but we assume that there exists at least an $i$ such that they are different:
\begin{equation}\label{eq:diffpdfassum}
	g_i \neq f_i, \quad \text{for some } i = 1, 2, \cdots, T.
\end{equation}
In this paper, we assume that the post-change law $(g_1, \cdots, g_T)$ is unknown. Further, there are $T$ families of distributions $\{\mathcal{P}_i\}_{i=1}^T$ such that
$$
g_i \in \mathcal{P}_i, \quad i=1,2, \dots, T.
$$
The families $\{\mathcal{P}_i\}_{i=1}^T$ are known to the decision maker. Below, we use the notation
$$
G = (g_1, g_2, \dots, g_T)
$$
to denote the post-change i.p.i.d. law.

Let $\tau$ be a stopping time for the process $\{X_n\}$, i.e., a positive integer-valued random variable such that the event $\{\tau \leq n\}$ belongs
to the $\sigma$-algebra generated by $X_1, \cdots, X_n$. In other words, whether or not $\tau \leq n$ is completely determined by the first $n$
observations. We declare that a change has occurred at the stopping time $\tau$. To find the best stopping rule to detect the change in distribution, we
need a performance criterion. Towards this end, we model the change point $\nu$ as a random variable with a prior distribution given by
$$
\pi_n = \Prob(\nu =n), \quad \text{ for } n = 1, 2, \cdots.
$$
For each $n \in \mathbb{N}$, we use $\Prob_n^G$ to denote the law of the observation process $\{X_n\}$ when the change occurs at $\nu=n$ and the post-change law is $G$. We use $\Expect_n^G$ to denote the corresponding expectation.
Using this notation, we define the average probability measure
$$
\Prob^{\pi,G} = \sum_{n=1}^\infty \pi_n \; \Prob_n^G.
$$
To capture a penalty for the false alarms,
in the event that the stopping time occurs before the change,
we use the probability of a false alarm defined as
$$
\Prob^{\pi,G}(\tau < \nu).
$$
Note that the probability of a false alarm
$\Prob^{\pi,G}(\tau < \nu)$ is not a function of the post-change law $G$. Hence, in the following, we suppress the mention of $G$ and refer to the probability of false alarm only by
$$
\Prob^{\pi}(\tau < \nu).
$$
To penalize the detection delay, we use the average detection delay given by
$$
\Expect^{\pi,G}\left[(\tau - \nu)^+\right],
$$
where $x^+ = \max\{x, 0\}$.

The optimization problem we are interested in solving is
\begin{equation}\label{eq:robustProb}
	\inf_{\tau \in \mathbf{C}_\alpha} \;\; \sup_{G: g_i \in \mathcal{P}_i, i \leq T} \; \Expect^{\pi,G}\left[(\tau - \nu)^+\right],
\end{equation}
where
$$
\mathbf{C}_\alpha = \left\{\tau:  \Prob^{\pi}(\tau < \nu) \leq \alpha\right\},
$$
and $\alpha$ is a given constraint on the probability of a false alarm.

In the case when the family of distributions $\{\mathcal{P}_i\}_{i=1}^T$ are singleton sets, i.e. when the post-change law is known and fixed $G$, a Lagrangian relaxation of this problem was investigated in \cite{bane-tit-2021}. Understanding the solution reported in \cite{bane-tit-2021} is fundamental to solving the robust problem in  \eqref{eq:robustProb}. In the next section, we discuss the solution provided in \cite{bane-tit-2021} and also its implication for the constrained version in \eqref{eq:robustProb}.

\subsection{Exactly and Asymptotically Optimal Solutions for Known Post-Change Law}
\label{sec:TITreview}
For known post-change law $G=(g_1, \dots, g_T)$ and geometrically distributed change point, it is shown in \cite{bane-tit-2021} that the exact optimal solution to a relaxed version of \eqref{eq:robustProb} is a stopping rule based on a periodic sequence of thresholds. It is also shown that it is sufficient to use only one threshold in the asymptotic regime of false alarm constraint $\alpha \to 0$. Furthermore, the assumption of a geometrically distributed change point can be relaxed in the asymptotic regime. In the rest of this section, we assume that $G$ is known and fixed.

\medspace
\medspace
\medspace
\subsubsection{Exactly Optimal Algorithm}
Let the change point $\nu$ be a geometric random variable:
$$
\Prob(\nu =n) =  (1-\rho)^{n-1} \rho, \quad \text{ for } n = 1, 2, \cdots.
$$
The relaxed version of \eqref{eq:robustProb} (for known $G$) is
\begin{equation}\label{eq:relaxedShiryProb}
	\inf_{\tau} \; \Expect^{\pi, G} \left[( \tau - \nu)^+\right] + \lambda_f \; \Prob^\pi(\tau < \nu),
\end{equation}
where $\lambda_f > 0$ is a penalty on the cost of false alarms.
Now, define $p_0=0$ and
\begin{equation}\label{eq:Shirpn}
	p_n = \Prob^{\pi, G}(\nu \leq n | X_1, \cdots, X_n), \text{ for } n \geq 1.
\end{equation}
Then, \eqref{eq:relaxedShiryProb} is equivalent to solving
\begin{equation}\label{eq:ShirPOMDP}
	\inf_{\tau} \; \Expect^{\pi, G}\left[ \sum_{n=0}^{\tau-1} p_n + \lambda_f (1-p_\tau)\right].
\end{equation}
The belief updated $p_n$ can be computed recursively using the following equations: $p_0=0$ and for $n \geq 1$,
\begin{equation}\label{eq:beliefUpdates}
	p_n = \frac{\tilde{p}_{n-1} \; g_n(X_n)}{\tilde{p}_{n-1} \; g_n(X_n) + (1-\tilde{p}_{n-1}) f_n(X_n)},
\end{equation}
where
$$
\tilde{p}_{n-1} = p_{n-1} + (1-p_{n-1}) \rho.
$$
Since these updates are not stationary, the problem cannot be solved using classical optimal stopping theory \cite{shir-opt-stop-book-1978} or dynamic programming \cite{bert-dyn-prog-book-2017}. However, the structure in \eqref{eq:beliefUpdates} repeats after every fixed time $T$. Motivated by this, in \cite{bane-tit-2021}, a control theory is developed for Markov decision processes with periodic transition and cost structures. This new control theory is then
used to solve the problem in \eqref{eq:ShirPOMDP}.

\medspace
\medspace
\medspace
\begin{theorem}[\hspace{-0.02cm}\cite{bane-tit-2021}]\label{thm:PeriodicThresOpt}
	There exist thresholds
	$A_1$, $A_2$, $\dots$, $A_T$, $A_i \geq 0, \forall i$, such that the  stopping rule
	\begin{equation}\label{eq:periodicThresOptAlgo}
		\tau^* = \inf \{n \geq 1: p_n \geq A_{(n \bmod T)}\},
	\end{equation}
	where $(n \bmod T)$ represents $n$ modulo $T$, is optimal for
	problem in \eqref{eq:ShirPOMDP}. These thresholds depend on the choice of $\lambda_f$.
\end{theorem}
\medspace
\medspace
\medspace

In fact, the solution given in \cite{bane-tit-2021} is valid for a more general change point problem in which separate delay and false alarm penalty is used for each time slot. We do not discuss it here.

\medspace
\medspace
\medspace

\subsubsection{Asymptotically Optimal Algorithm}
For large values of $T$, which can easily be more than a million for certain applications, it is computationally not feasible to store $T$ different values of thresholds. Thus, it is of interest to see if a single-threshold algorithm is optimal. It is shown in \cite{bane-tit-2021} that periodic threshold algorithms are strictly optimal. However, it is shown in \cite{bane-tit-2021} that a single-threshold test is asymptotically optimal in the regime of low probability of false alarms. We discuss this result below. 

Let there exist $d\geq 0$ such that
\begin{equation}\label{eq:priortail}
	\lim_{n \to \infty} \frac{\log \Prob(\nu > n)}{n} = -d.
\end{equation}
If $\pi = \text{Geom}(\rho)$, then $d = |\log(1-\rho)|$.
Further, let
\begin{equation}\label{eq:KLnumber}
	I = \frac{1}{T}\sum_{i=1}^T D(g_i \; \| \; f_i),
\end{equation}
where $D(g_i \; \| \; f_i)$ is the Kullback-Leibler divergence between the densities $g_i$ and $f_i$:
$$
D(g_i \; \| \; f_i) = \int g_i(x) \log \frac{g_i(x)}{f_i(x)} dx. 
$$
The following theorem is proved in \cite{bane-tit-2021}. 


\begin{theorem}[\hspace{-0.02cm}\cite{bane-tit-2021}]\label{thm:LB}
	Let the information number $I$ be as defined in \eqref{eq:KLnumber} and satisfy $0 < I < \infty$. Also, let $d$ be as in \eqref{eq:priortail}.
	Then, with
	$$
	A_1 = A_2 = \dots = A_T = 1-\alpha,
	$$
	$\tau^* \in \mathbf{C}_\alpha$, and \vspace{-0.2cm}
	\begin{equation}
		\begin{split}
			\Expect^{\pi, G}\left[ (\tau^* - \nu)^+\right]	&= \inf_{\tau \in \mathbf{C}_\alpha}
			\Expect^{\pi, G}\left[ (\tau - \nu)^+\right](1+o(1))\\
			&= \frac{|\log \alpha| }{I + d}(1+o(1)), \quad \text{ as } \alpha \to 0.
		\end{split}
	\end{equation}
	Here $o(1) \to 0$ as $\alpha \to 0$.
\end{theorem}
\medspace
\medspace
\medspace

\subsubsection{Solution to the Constraint Version of the Problem}
We now argue that, just as in the classical case, the relaxed version of the problem \eqref{eq:ShirPOMDP} can be used to provide a solution to the constraint version of the problem \eqref{eq:robustProb}. We provide proof for completeness.

\medspace
\medspace
\medspace
\begin{lemma}
\label{lem:langtoconst}
	If $\alpha$ is a value of the probability of false alarm achievable by the optimal stopping rule $\tau^*$ in \eqref{eq:ShirPOMDP}, then $\tau^*$ is also optimal for the constraint problem \eqref{eq:robustProb} for this $\alpha$.
\end{lemma}

\medspace
\medspace
\medspace
\begin{proof}
	By Theorem~\ref{thm:PeriodicThresOpt}, we have
	\begin{equation}\label{eq:relaxedShiryProb1}
		\begin{split}
			\Expect^{\pi, G} \left[( \tau^* - \nu)^+\right] &+ \lambda_f \; \Prob^\pi(\tau^* < \nu) 
			\leq \Expect^{\pi, G} \left[( \tau - \nu)^+\right] + \lambda_f \; \Prob^\pi(\tau < \nu).
		\end{split}
	\end{equation}
	If $\Prob^\pi(\tau^* < \nu) =\alpha$ and $\Prob^\pi(\tau < \nu) \leq \alpha$, then
	\begin{equation}\label{eq:relaxedShiryProb2}
		\begin{split}
			\Expect^{\pi, G} &\left[( \tau^* - \nu)^+\right] + \lambda_f \; \Prob^\pi(\tau^* < \nu) 
			=\Expect^{\pi, G} \left[( \tau^* - \nu)^+\right] + \lambda_f \; \alpha \\
			&\leq \Expect^{\pi, G} \left[( \tau - \nu)^+\right] + \lambda_f \; \Prob^\pi(\tau < \nu) 
			\leq \Expect^{\pi, G} \left[( \tau - \nu)^+\right] + \lambda_f \; \alpha.
		\end{split}
	\end{equation}
	Canceling $\lambda_f \; \alpha$ from both sides we get
	$$
	\Expect^{\pi, G} \left[( \tau^* - \nu)^+\right] \leq \Expect^{\pi, G} \left[( \tau - \nu)^+\right].
	$$
\end{proof}

The following lemma guarantees that a wide range of probability of false alarm $\alpha$ is achievable by the optimal stopping rule $\tau^*$.

\medspace
\medspace
\medspace

\begin{lemma}
	As we increase $\lambda_f \to \infty$ in \eqref{eq:ShirPOMDP}, the probability of false alarm achieved by the optimal stopping rule $\tau^*$ goes to zero.
\end{lemma}
\medspace
\medspace
\medspace
\begin{proof}
	As $\lambda_f \to \infty$, if the probability of false alarm for $\tau^*$ stays bounded away from zero, then the Bayesian risk
	$$
	\Expect^{\pi, G} \left[( \tau^* - \nu)^+\right] + \lambda_f \; \Prob^\pi(\tau^* < \nu)
	$$
	would diverge to infinity. This will contradict the fact that $\tau^*$ is optimal because we can get a smaller risk at large enough $\lambda_f$ by stopping at a large enough deterministic time.
\end{proof}

\subsection{Optimal Robust Algorithm for Unknown Post-Change Law}
We now assume that the post-change law $G$ is unknown and provide the optimal solution to \eqref{eq:robustProb} under assumptions on the families of post-change laws $\{\mathcal{P}_i\}_{i=1}^T$. Specifically, we extend the results in \cite{unni-etal-ieeeit-2011} for i.i.d. processes to i.p.i.d. processes.
We assume in the rest of this section that all densities involved are equivalent to each other (absolutely continuous with respect to each other). Also, we assume that the change point $\nu$ is a  geometrically distributed random variable.

To state the assumptions on $\{\mathcal{P}_i\}_{i=1}^T$, we need some defintions. We say that a random variable $Z_2$ is stochastically larger than another random variable $Z_1$ if
$$
\Prob(Z_2 \geq t) \geq \Prob(Z_1 \geq t), \quad \text{for all } t \in \mathbb{R}.
$$
We use the notation
$$
Z_2 \succ Z_1.
$$
If $\mathcal{L}_{Z_2}$ and $\mathcal{L}_{Z_1}$ are the probability laws of $Z_2$ and $Z_1$, then we also use the notation
$$
\mathcal{L}_{Z_2} \succ \mathcal{L}_{Z_1}.
$$
We now introduce the notion of stochastic boundedness in i.p.i.d. processes. In the following, we use
$$
\mathcal{L}(\phi(X), g)
$$
to denote the law of some function $\phi(X)$ of the random variable $X$, when the variable $X$ has density $g$.

\medspace
\medspace
\medspace
\medspace
\medspace

\begin{definition}[Stochastic Boundedness in i.p.i.d. Processes; Least Favorable Law]
	We say that the family $\{\mathcal{P}_i\}_{i=1}^T$ is stochastically bounded by the i.p.i.d. law
	$$
	\bar{G}=(\bar{g}_1, \bar{g}_2, \dots, \bar{g}_T),
	$$
	and call $\bar{G}$ the least favorable law (LFL), if
	$$
	\bar{g}_i \in \mathcal{P}_i, \quad i=1,2, \dots, T,
	$$
	and
	\begin{equation}
 \label{eq:stocbounded}
		\begin{split}
			\mathcal{L}\left(\log \frac{\bar{g}_i(X_i)}{f_i(X_i)}, g_i\right) &\succ 	\mathcal{L}\left(\log \frac{\bar{g}_i(X_i)}{f_i(X_i)},
			\bar{g}_i\right), \quad \text{for all } \; g_i \in \mathcal{P}_i, \quad i=1,2, \dots, T.
		\end{split}
	\end{equation}
\end{definition}

\medspace
\medspace
\medspace
\medspace
\medspace

Consider the stopping rule $\tau^*$ designed using the LFL $\bar{G}=(\bar{g}_1, \bar{g}_2, \dots, \bar{g}_T)$:
\begin{equation}\label{eq:LFLshir}
	\bar{\tau}^* = \inf \{n \geq 1: \bar{p}_n \geq A_{(n \bmod T)}\},
\end{equation}
where $\bar{p}_0=0$, and
\begin{equation}\label{eq:LFPbelief}
	\bar{p}_n = \frac{\tilde{p}_{n-1} \; \bar{g}_n(X_n)}{\tilde{p}_{n-1} \; \bar{g}_n(X_n) + (1-\tilde{p}_{n-1}) f_n(X_n)},
\end{equation}
where
$$
\tilde{p}_{n-1} = \bar{p}_{n-1} + (1-\bar{p}_{n-1}) \rho.
$$

We now state our main result on robust quickest change detection in i.p.i.d. processes. 

\medspace
\medspace
\medspace
\medspace

\begin{theorem}
	Suppose the following conditions hold:
	\begin{enumerate}
		\item 	The family $\{\mathcal{P}_i\}_{i=1}^T$ be stochastically bounded by the i.p.i.d. law
		$$
		\bar{G}=(\bar{g}_1, \bar{g}_2, \dots, \bar{g}_T).
		$$
		\item Let $\alpha \in [0,1]$ be a constraint such that
		$$
		\Prob^\pi(\bar{\tau}^* < \nu) = \alpha,
		$$
		where $\bar{\tau}^*$ is the optimal rule designed using the LFL \eqref{eq:LFLshir}.
		\item All likelihood ratio functions involved are continuous.
		\item The change point $\nu$ is geometrically distributed.
	\end{enumerate}
	Then, the stopping rule $\bar{\tau}^*$ in \eqref{eq:LFLshir} designed using the LFL is optimal for the robust constraint problem in \eqref{eq:robustProb}.
	
\end{theorem}

\medspace
\medspace
\medspace

\begin{proof}
	The key step in the proof is to show that for each $k \in \mathbb{N}$,
	\begin{equation}\label{eq:keystep}
		\begin{split}
			\Expect_k^{\bar{G}}&\left[(\bar{\tau}^* - k)^+ |\mathcal{F}_{k-1}\right] \succ \Expect_k^{G}\left[(\bar{\tau}^* - k)^+ | \mathcal{F}_{k-1}\right], \\
			&\; \text{ for all } G=(g_1, \dots g_T): g_i \in \mathcal{P}_i, \; i \leq T,
		\end{split}
	\end{equation}
	where $\mathcal{F}_{k-1}$ is the sigma algebra generated by observations $X_1, \dots, X_{k-1}$.
	If \eqref{eq:keystep} is true then we have for each $k \in \mathbb{N}$,
	\begin{equation}
		\begin{split}
			\Expect_k^{\bar{G}}&\left[(\bar{\tau}^* - k)^+\right] \geq \Expect_k^{G}\left[(\bar{\tau}^* - k)^+ \right] \\
			&\; \text{ for all } G=(g_1, \dots g_T): g_i \in \mathcal{P}_i, \; i \leq T.
		\end{split}
	\end{equation}

	Averaging over the prior on the change point, we get
	\begin{equation}
		\begin{split}
			&\Expect^{\pi,\bar{G}}\left[(\bar{\tau}^* - \nu)^+\right] =  \sum_k \pi_k \Expect_k^{\bar{G}}\left[(\bar{\tau}^* - k)^+\right] 
			\geq \sum_k \pi_k \Expect_k^{G}\left[(\bar{\tau}^* - k)^+ \right] =	\Expect^{\pi,G}\left[(\bar{\tau}^* - \nu)^+ \right], \\
			&\; \text{ for all } G=(g_1, \dots g_T): g_i \in \mathcal{P}_i, \; i \leq T.
		\end{split}
	\end{equation}
	The last equation gives
	\begin{equation}
		\begin{split}
			\Expect^{\pi,\bar{G}}&\left[(\bar{\tau}^* - \nu)^+\right]  \geq \Expect^{\pi,G}\left[(\bar{\tau}^* - \nu)^+ \right], \\
			&\; \text{ for all } G=(g_1, \dots g_T): g_i \in \mathcal{P}_i, \; i \leq T.
		\end{split}
	\end{equation}
	This implies that
	\begin{equation}
		\begin{split}
			\Expect^{\pi,\bar{G}}&\left[(\bar{\tau}^* - \nu)^+\right]  = \sup_{G: g_i \in \mathcal{P}_i, i \leq T} \Expect^{\pi,G}\left[(\bar{\tau}^* - \nu)^+ \right],
		\end{split}
	\end{equation}
 where we have equality because the law $\bar{G}$ belongs to the family considered on the right. 
	Now, if $\tau$ is any stopping rule satisfying the probability of false alarm constraint of $\alpha$, then since $\bar{\tau}^*$ is the optimal test for the LFL $\bar{G}$, we have (see Theorem~\ref{thm:PeriodicThresOpt} and Lemma~\ref{lem:langtoconst})
	\begin{equation}
		\begin{split}
			\sup_{G: g_i \in \mathcal{P}_i, i \leq T} \Expect^{\pi,G}\left[(\tau - \nu)^+ \right] &\geq \Expect^{\pi,\bar{G}}\left[(\tau - \nu)^+\right]
			\geq
			\Expect^{\pi,\bar{G}}\left[(\bar{\tau}^* - \nu)^+\right]  \\
			&= \sup_{G: g_i \in \mathcal{P}_i, i \leq T} \Expect^{\pi,G}\left[(\bar{\tau}^* - \nu)^+ \right].
		\end{split}
	\end{equation}
	The last equation proves the robust optimality of the stopping rule $\bar{\tau}^*$ for the problem in \eqref{eq:robustProb}.
	
	We now prove the key step \eqref{eq:keystep}. Towards this end, we prove
	that for every integer $N\geq 0$,
	\begin{equation}\label{eq:keystep2}
		\begin{split}
			\Prob_k^{\bar{G}}\left[(\bar{\tau}^* - k)^+ > N |\mathcal{F}_{k-1}\right] &\geq \Prob_k^{G}\left[(\bar{\tau}^* - k)^+ > N| \mathcal{F}_{k-1}\right], \\
			&\quad \quad \quad \; \text{ for all } G=(g_1, \dots g_T): g_i \in \mathcal{P}_i, \; i \leq T.
		\end{split}
	\end{equation}
	This is trivially true for $N=0$ since event $\{(\bar{\tau}^* - k)^+ > 0\}$ is $\mathcal{F}_{k-1}$-measurable. So we only prove it for $N \geq 1$. 
 Towards this end, we first have
	\begin{equation}\label{eq:temp1}
		\begin{split}
			\Prob_k^{\bar{G}}&\left[(\bar{\tau}^* - k)^+ \leq N |\mathcal{F}_{k-1}\right] 
			= \Prob_k^{\bar{G}}\left[\bar{\tau}^* \leq k+N |\mathcal{F}_{k-1}\right] \\
			&\quad = \Prob_k^{\bar{G}}\left[f(h_1(X_1), h_2(X_2), \dots, h_{k+N}(X_{k+N})) \; \geq \; 0 \; | \; \mathcal{F}_{k-1}\right],
		\end{split}
	\end{equation}
	where 
 $$
 h_i(x) = \log \frac{\bar{g}_i(x)}{f_i(x)}, 
 $$ the function $f(z_1, z_2, \dots, z_{N})$ is given by
	\begin{equation}
		\begin{split}
			f&(z_1, z_2, \dots, z_N) 
			= \max_{1 \leq n \leq N} \left(\sum_{i=1}^n (1-\rho)^{k-1}\rho \; \exp \left(\sum_{i=k}^n z_i\right)-B_n\right),
		\end{split}
	\end{equation}
 and 
 $$
 B_n = \frac{A_n}{1-A_n} (1-\rho)^n.
 $$ Now recall from \eqref{eq:stocbounded} that 
 \begin{equation}
 		\begin{split}
			\mathcal{L}\left(h_i(X_i), g_i\right) &\succ 	\mathcal{L}\left(h_i(X_i),
			\bar{g}_i\right), \quad \text{for all } \; g_i \in \mathcal{P}_i, \quad i=1,2, \dots, T.
		\end{split}
	\end{equation}
 Since the function $f$ is continuous (being the maximum of continuous functions) and non-decreasing in each of its arguments, Lemma III.1 in \cite{unni-etal-ieeeit-2011} implies that
	\begin{equation}\label{eq:temp2}
		\begin{split}
			\Prob_k^{\bar{G}}&\left[f(h_1(X_1), h_2(X_2), \dots, h_{k+N}(X_{k+N})) \; \geq \; 0 \; |\; \mathcal{F}_{k-1}\right] 
			\\
   &\leq \Prob_k^{G}\left[f(h_1(X_1), h_2(X_2), \dots, h_{k+N}(X_{k+N})) \; \geq \; 0 \; |\; \mathcal{F}_{k-1}\right], \\
   &\quad \quad \text{ for all } G=(g_1, \dots g_T): g_i \in \mathcal{P}_i, \; i \leq T. 
		\end{split}
	\end{equation}
	Equations \eqref{eq:temp1} and \eqref{eq:temp2} combined gives
	\begin{equation}\label{eq:temp3}
		\begin{split}
			\Prob_k^{\bar{G}}\left[(\bar{\tau}^* - k)^+ \leq N \; | \; \mathcal{F}_{k-1}\right] 
			 &= \Prob_k^{\bar{G}}\left[\bar{\tau}^* \leq k+N \; |\; \mathcal{F}_{k-1}\right] \\
			& = \Prob_k^{\bar{G}}\left[f(h_1(X_1), h_2(X_2), \dots, h_{k+N}(X_{k+N})) \geq 0 \; |\; \mathcal{F}_{k-1}\right] \\
			&\leq \Prob_k^{G}\left[f(h_1(X_1), h_2(X_2), \dots, h_{k+N}(X_{k+N})) \geq 0 \; | \; \mathcal{F}_{k-1}\right] \\
   &= \Prob_k^{G}\left[\bar{\tau}^* \leq k+N \; |\; \mathcal{F}_{k-1}\right] \\
			&=\Prob_k^{{G}}\left[(\bar{\tau}^* - k)^+ \leq N \; |\; \mathcal{F}_{k-1}\right], \\
   &\quad \quad \quad \text{ for all } G=(g_1, \dots g_T): g_i \in \mathcal{P}_i, \; i \leq T.
		\end{split}
	\end{equation}
	This proves \eqref{eq:keystep2} and hence \eqref{eq:keystep}.
\end{proof}

\medspace
\medspace
\medspace

\section{Quickest Joint Detection and Classification}
\label{sec:QCDFaultIsolation}

\subsection{Joint Detection and Classification Formulation}
We assume that in a normal regime, the data can be modeled as an i.p.i.d. process with the law $(g_1^{(0)}, \cdots, g_T^{(0)})$. At some point in time $\nu$,
the law of the i.p.i.d. process is governed not by the densities $(g_1^{(0)}, \cdots, g_T^{(0)})$,
but by one of the densities $(g_1^{(\ell)}, \cdots, g_T^{(\ell)} )$,  $\ell=1,2,\dots, M$, with
$$
g_{n+T}^{(\ell)} = g_{n}^{(\ell)},  \quad \forall n \geq 1, \quad \ell=1,2,\dots, M.
$$
Specifically, at the time point $\nu$, the distribution of the random variables change from $\{g_n^{(0)}\}$ to $\{g_n^{(\ell)} \}$: for some $\ell=1,2,\dots, M$,
\begin{equation}\label{eq:changepointmodelfault}
	X_n \sim
	\begin{cases}
		g_n^{(0)}, &\quad \forall n < \nu, \\
		g_n^{(\ell)}  &\quad \forall n \geq \nu.
	\end{cases}
\end{equation}

We want to detect the change described in \eqref{eq:changepointmodelfault} as quickly as possible, subject to a constraint on the rate of false alarms and on the probability of misclassification. Mathematically, we are looking for a pair $(\tau, \delta)$, where $\tau$ is stopping time, i.e.,
$$
\{\tau \leq n\} \in \sigma(X_1, X_2, \dots, X_n),
$$
and $\delta$ is a decision rule, i.e., a map such that
$$
\delta(X_1, X_2, \dots, X_\tau) \in \{1,2,\dots, M\}.
$$
Let $\Prob_\nu^{(\ell)}$ denote the
probability law of the process $\{X_n\}$ when the change occurs at time $\nu$ and the post-change law is $(g_1^{(\ell)}, \cdots, g_T^{(\ell)} )$. We let $\Expect_\nu^{(\ell)}$ denote the corresponding expectation. When there is no change,
we use the notation $\Expect_\infty$.
The problem of interest is as follows \cite{lai2000faultiso}:
\begin{equation}\label{eq:Lorden}
	\begin{split}
		\min_{\tau, \delta} &\;\;\max_{1 \leq \ell \leq M} \sup_{\nu\geq 1} \; \text{ess} \sup \Expect_\nu^{(\ell)}  [(\tau - \nu + 1)^+| X_1, \cdots, X_{\nu-1}],\\
		\text{subj. to}& \; \; \;\Expect_\infty[\tau] \geq \beta, \\
		\text{and }& \;\;\;\Prob_1^{(\ell)}[\tau < \infty, \delta\neq \ell ] \leq a_\beta \; \Expect_1^{(\ell)}[\tau], \; \quad \ell=1,2,\dots, M, \\
		& \quad \quad \text{where } \log a_\beta^{-1} \sim \log \beta, \; \text{ as } \beta \to \infty.
	\end{split}
\end{equation}
Here $\text{ess} \sup$ is the essential supremum of the random variable $\Expect_\nu^{(\ell)}  [(\tau - \nu + 1)^+| X_1, \cdots, X_{\nu-1}]$, i.e., the smallest constant dominating the random variable with probability one. Here and below, for two functions $h(\beta)$ and $f(\beta)$ of $\beta$, we use $f(\beta) \sim h(\beta)$, as $\beta \to \infty$, to denote that the ratio of the two functions goes to $1$ in the limit.
Further motivation for this and other problem formulations for change point detection and isolation can be found in the literature \cite{tart-niki-bass-2014}, \cite{lai2000faultiso}, \cite{niki-ieeetit-2003}.

\subsection{Algorithm for Detection when $M=1$}
When $M=1$, i.e., when there is only one post-change i.p.i.d. law, then an algorithm that is asymptotically optimal for detecting a change in the distribution is the periodic-CUSUM algorithm proposed in \cite{bane-icassp-2019}. In this algorithm, we compute the sequence of statistics
\begin{equation}\label{eq:PeriodicCUSUM}
	W_{n+1} = W_n^{+} + \log \frac{g_{n+1}^{(1)}(X_{n+1})}{g_{n+1}^{(0)}(X_{n+1})}
\end{equation}
and raise an alarm as soon as the statistic is above a threshold $A$:
\begin{equation}\label{eq:PeriodicCUSUMstop}
	\tau_c = \inf \{n \geq 1: W_n \geq A\}.
\end{equation}

Define
\begin{equation}\label{eq:KLnumberfault}
	I_{10} =  \frac{1}{T}\sum_{i=1}^T D(g_i^{(1)} \; \| \; g_i^{(0)}),
\end{equation}
where $D(g_i^{(1)} \; \| \; g_i^{(0)})$ is the Kullback-Leibler divergence between the densities $g_i^{(1)} $ and $g_i^{(0)}$. Then, the following result is proved in \cite{bane-icassp-2019}.

\begin{theorem}[\cite{bane-icassp-2019}]\label{thm:LB_fault}
	Let the information number $I_{10}$ as defined in \eqref{eq:KLnumberfault} satisfy $0 < I_{10} < \infty$. Then, with $A=\log \beta$,
	$$\Expect_\infty[\tau_c] \geq \beta,
	$$
	and as $\beta \to \infty$,
	\begin{equation}\label{eq:UB}
		\begin{split}
			&\sup_{\nu\geq 1} \; \text{ess} \sup \Expect_\nu [(\tau_c - \nu +1^+| X_1, \cdots, X_{\nu-1}] \\
			&\sim\inf_{\tau: \Expect_\infty[\tau] \geq \beta} \; \sup_{\nu\geq 1} \; \text{ess} \sup \Expect_\nu [(\tau - \nu+1)^+ | X_1, \cdots, X_{\nu-1}]\\
			&\sim  \frac{\log \beta}{I_{10}}.
		\end{split}
	\end{equation}
\end{theorem}
Thus, the periodic-CUSUM algorithm is asymptotically optimal for detecting a change in the distribution, as the false alarm constraint $\beta \to \infty$. Further, since the set of pre- and post-change densities $(g_1^{(0)}, \cdots, g_T^{(0)} )$ and $(g_1^{(1)}, \cdots, g_T^{(1)} )$ are finite,
the recursion in \eqref{eq:PeriodicCUSUM} can be computed using finite memory needed to store these $2T$ densities.

\subsection{Algorithm for Joint Detection and Classification}\label{sec:FalseIso}
When the possible number of post-change distributions $M > 1$ and when we are also interested in accurately classifying the true post-change law, the periodic-CUSUM algorithm is not sufficient. We now propose an algorithm that can perform joint detection and classification.

For $\ell = 1, \dots, M$, define the stopping times
\begin{equation}\label{eq:tauell}
	\begin{split}
		\tau_\ell &= \inf\left\{n \geq 1: \max_{1 \leq k \leq n} \; \min_{0 \leq m \leq M, m \neq \ell} \;  \sum_{i=k}^n \log \frac{g_{i}^{(\ell)}(X_{i})}{g_{i}^{(m)}(X_{i})} \geq A \right\}.
	\end{split}
\end{equation}
The stopping time and decision rule for our detection-classification problem is defined as follows:
\begin{equation}
	\begin{split}
		\tau_{dc} &= \min_{1 \leq \ell \leq M} \; \tau_\ell, \\
		\delta_{dc} &= \arg \min_{1 \leq \ell \leq M} \tau_\ell.
	\end{split}
\end{equation}
A window-limited version of the above algorithm is obtained by replacing each $\tau_\ell$ in \eqref{eq:tauell} by
\begin{equation}\label{eq:windowlimited}
	\begin{split}
		\tilde{\tau}_\ell &= \inf\left\{n: \max_{n-L_\beta \leq k \leq n} \; \min_{0 \leq m \leq M, m \neq \ell} \; \sum_{i=k}^n \log \frac{g_{i}^{(\ell)}(X_{i})}{g_{i}^{(m)}(X_{i})} \geq A \right\}
	\end{split}
\end{equation}
for an appropriate choice of window $L_\beta$ (to be specified in the theorem below). 

For $1 \leq \ell \leq M$ and $0 \leq m \leq M, \; m \neq \ell$, define
\begin{equation}\label{eq:KLnumberellm}
	I_{\ell m} = \frac{1}{T}\sum_{i=1}^T D(g_i^{(\ell)} \; \| \; g_i^{(m)}),
\end{equation}
and
\begin{equation}\label{eq:KLstar}
	I^* = \min_{1 \leq \ell \leq M}\; \min_{0 \leq m \leq M, m \neq \ell}  \; I_{\ell m}.
\end{equation}
Recall that we are looking for $(\tau, \delta)$ such that
\begin{equation}\label{eq:MFA}
	\begin{split}
		\Expect_\infty[\tau] \geq \beta (1+o(1)), \text{ as } \beta \to \infty,
	\end{split}
\end{equation}
and
\begin{equation}\label{eq:faultIsolProb}
	\begin{split}
		\Prob_1^{(\ell)}&[\tau < \infty, \delta\neq \ell ] \leq a_\beta \; \Expect_1^{(\ell)}[\tau], \; \quad \ell=1,2,\dots, M, \\
		&\quad \quad \text{where } \log a_\beta^{-1} \sim \log \beta, \; \text{ as } \beta \to \infty.
	\end{split}
\end{equation}
Let
\begin{equation}
	C_\beta = \{(\tau, \delta): \text{conditions in \eqref{eq:MFA}} \text{ and } \eqref{eq:faultIsolProb} \text{ hold}\}.
\end{equation}

\begin{theorem}\label{thm:isolation}
	Let the information number $I^*$ be as defined in \eqref{eq:KLstar} and satisfy $0 < I^* < \infty$. Then, with $A=\log 4M \beta$,
	$$
	(\tau_{dc}, \delta_{dc}) \in C_\beta.
	$$
	Also,
	\begin{equation}\label{eq:UBiso}
		\begin{split}
			&\max_{1 \leq \ell \leq M} \sup_{\nu\geq 1} \; \text{ess} \sup \Expect_\nu^{(\ell)}  [(\tau_{dc} - \nu + 1)^+| X_1, \cdots, X_{\nu-1}]\\
			&\sim \inf_{(\tau, \delta) \in C_\beta} \max_{1 \leq \ell \leq M} \sup_{\nu\geq 1} \; \text{ess} \sup \Expect_\nu^{(\ell)}  [(\tau - \nu + 1)^+| X_1, \cdots, X_{\nu-1}]\\
			&\sim  \frac{\log \beta}{I^*}, \; \text{ as } \beta \to \infty.
		\end{split}
	\end{equation}
	Finally, the window-limited version of the test \eqref{eq:windowlimited} also satisfies the same asymptotic optimality property as long as
	$$
	\lim \inf_{\beta \to \infty} \frac{L_\beta}{\log \beta} > \frac{1}{I^*}.
	$$
 This condition is satisfied, for example, by 
 $$
 L_\beta = \frac{\log \beta}{I^*}(1+ \epsilon)
 $$
 for any fixed $\epsilon > 0$.
\end{theorem}

%
%



\begin{proof}
	For $1 \leq \ell \leq M$ and $0 \leq m \leq M, \; m \neq \ell$, define
	$$
	Z_i(\ell, m) = \log \frac{g_{i}^{(\ell)}(X_{i})}{g_{i}^{(m)}(X_{i})}
	$$	
	to be the log-likelihood ratio at time $i$.
 In the rest of the proof, to write compact equations, we use $X_1^{\nu-1}$ to denote the vector
	$$
	X_1^{\nu-1} = (X_1, X_2, \dots, X_{\nu-1}).
	$$
	For each $1 \leq \ell \leq M$ and $0 \leq m \leq M, \; m \neq \ell$, we first show that the sequence $\{Z_i(\ell, m)\}$ satisfies the following statement:
	\begin{equation}\label{eq:thm1_1}
		\begin{split}
			\sup_{\nu \geq 1} \esssup \Prob_\nu^{(\ell)} & \left(  \max_{t \leq n} \sum_{i=\nu}^{\nu+t} Z_i(\ell,m) \geq I_{\ell m}(1+\delta)n \; \bigg| \;  X_1^{\nu-1}\right) \\
			& \quad \quad \quad \to 0, \text{ as } n \to \infty, \quad \forall \delta > 0,
		\end{split}
	\end{equation}
	where $I_{\ell m}$ is as defined in \eqref{eq:KLnumberellm}.
	
	Towards proving \eqref{eq:thm1_1}, note that as $n \to \infty$
	\begin{equation}\label{eq:thm1_2}
		\begin{split}
			\frac{1}{n}\sum_{i=\nu}^{\nu+n} Z_i(\ell, m)  \to I_{\ell m}, \quad \text{a.s. } \; \Prob_\nu^{(\ell)}, \; \; \forall \nu \geq 1.
		\end{split}
	\end{equation}
	The above display is true because of the i.p.i.d. nature of the observation process. This implies that as $n \to \infty$
	\begin{equation}\label{eq:thm1_3}
		\begin{split}
			\max_{t \leq n} \frac{1}{n}\sum_{i=\nu}^{\nu+t} Z_i(\ell, m) \to I_{\ell m}, \quad \text{a.s. }  \; \Prob_\nu^{(\ell)}, \; \; \forall \nu \geq 1.
		\end{split}
	\end{equation}
	To show this, note that
	\begin{equation}\label{eq:thm1_4}
		\begin{split}
			\max_{t \leq n}& \frac{1}{n}\sum_{i=\nu}^{\nu+t} Z_i(\ell, m)  = \max \left\{ \max_{t \leq n-1} \frac{1}{n}\sum_{i=\nu}^{\nu+t} Z_i(\ell, m), \; \;  \frac{1}{n}\sum_{i=\nu}^{\nu+n} Z_i(\ell, m)\right\}.
		\end{split}
	\end{equation}
	For a fixed $\epsilon > 0$, because of \eqref{eq:thm1_2}, the LHS in \eqref{eq:thm1_3} is greater than $I_{\ell m}(1-\epsilon)$ for $n$ large enough. Also, let the maximum on the LHS be achieved at a point $k_n$,
	then
	$$
	\max_{t \leq n} \frac{1}{n}\sum_{i=\nu}^{\nu+t} Z_i(\ell, m) = \frac{1}{n}\sum_{i=\nu}^{\nu+k_n} Z_i(\ell, m)  = \frac{k_n}{n} \frac{1}{k_n}\sum_{i=\nu}^{\nu+k_n} Z_i(\ell, m).
	$$
	Now $k_n$ cannot be bounded because the left-hand side in the above equation is lower bounded by $I_{\ell m}(1-\epsilon)$, and because of the presence of $n$ in the denominator on the right-hand side of the above equation. This implies $k_n > i$, for any fixed $i$, and $k_n \to \infty$. Thus, $\frac{1}{k_n}\sum_{i=\nu}^{\nu+k_n} Z_i(\ell, m) \to I_{\ell m }$. Since $k_n/n \leq 1$, we have that the LHS in \eqref{eq:thm1_3} is less than $I_{\ell m }(1+\epsilon)$, for $n$ large enough. This proves
	\eqref{eq:thm1_3}.
	To prove \eqref{eq:thm1_1}, note that due to the i.p.i.d. nature of the processes
	\begin{equation}\label{eq:thm1_5}
		\begin{split}
			\sup_{\nu \geq 1} \esssup \Prob_\nu^{(\ell)} & \left(  \max_{t \leq n} \sum_{i=\nu}^{\nu+t} Z_i(\ell,m) \geq I_{\ell m}(1+\delta)n \; \bigg| \;  X_1^{\nu-1}\right) \\
			=\sup_{1 \leq \nu \leq T} \Prob_\nu^{(\ell)} & \left(  \max_{t \leq n} \sum_{i=\nu}^{\nu+t} Z_i(\ell,m) \geq I_{\ell m}(1+\delta)n\right).
		\end{split}
	\end{equation}
	The right-hand side goes to zero because of \eqref{eq:thm1_3} and because the maximum on the right-hand side in \eqref{eq:thm1_5} is over only finitely many terms.
	
	Next, we show that the sequence $\{Z_i(\ell, m)\}$, for each $1 \leq \ell \leq M$ and $0 \leq m \leq M, \; m \neq \ell$, satisfies the following statement:
	\begin{equation}\label{eq:thm2_1}
		\begin{split}
			\lim_{n \to \infty} \sup_{k \geq \nu \geq 1} \esssup \; \Prob_\nu^{(\ell)} & \left(  \frac{1}{n} \sum_{i=k}^{k+n} Z_i(\ell, m) \leq I_{\ell m} - \delta \; \bigg| \; X_1^{k-1}\right) \\
			& = 0, \quad \forall \delta > 0.
		\end{split}
	\end{equation}
	To prove \eqref{eq:thm2_1}, note that due to the i.p.i.d nature of the process we have
	\begin{equation}\label{eq:thm2_2}
		\begin{split}
			\sup_{k \geq \nu \geq 1} &\esssup \; \Prob_\nu^{(\ell)}  \left(  \frac{1}{n} \sum_{i=k}^{k+n} Z_i(\ell,m) \leq I_{\ell m} - \delta \; \bigg| \; X_1^{k-1}\right) \\
   &=\sup_{k \geq \nu \geq 1}  \Prob_\nu^{(\ell)}  \left(  \frac{1}{n} \sum_{i=k}^{k+n} Z_i(\ell,m) \leq I_{\ell m} - \delta \; \right) \\
			& = \sup_{\nu + T \geq k \geq \nu \geq 1} \Prob_\nu^{(\ell)}  \left(  \frac{1}{n} \sum_{i=k}^{k+n} Z_i(\ell,m) \leq I_{\ell m } - \delta \right) \\
			& = \max_{1 \leq \nu \leq T} \max_{\nu \leq k \leq \nu+T} \Prob_\nu^{(\ell)}  \left(  \frac{1}{n} \sum_{i=k}^{k+n} Z_i(\ell, m ) \leq I_{\ell m} - \delta \right).
		\end{split}
	\end{equation}
	The right-hand side of the above equation goes to zero as $n\to \infty$ for any $\delta$ because of \eqref{eq:thm1_2} and also because of the finite number of maximizations. The theorem now follows from Theorem 4 of \cite{lai2000faultiso} because Conditions A1 and A2 from \cite{lai2000faultiso} are satisfied.
\end{proof}

\begin{remark}
	The algorithm and optimality are easily extended to multistream data as well, where there are a finite number of streams of observations, and only one stream is affected after the change. The goal is to detect the change and also to identify the affected stream with low probability.  
\end{remark}

\section{Multislot Quickest Change Detection}
\label{sec:multislot}


Let $(f_1, \cdots, f_T)$ and $(g_1, \cdots, g_T)$ represent the laws of two i.p.i.d processes with $f_i \neq g_i$, $\forall i$. We assume that the second-order moments of log-likelihood ratios are finite and positive:
\begin{equation}\label{eq:finitemeancond}\vspace{-0.2cm}
	\begin{split}
		\text{(M1)} \quad &0< \Expect_1 \left(\left|\log \frac{g_i(X_i)}{f_i(X_i)}\right|\right)  < \infty, \; i=1,2, \cdots, T \\
		\text{(V1)} \quad &0< \Expect_1 \left(\log \frac{g_i(X_i)}{f_i(X_i)}\right)^2  < \infty, \; i=1,2, \cdots, T.
		\quad \end{split}
\end{equation}
Here $\Expect_1$ denotes the expectation when the change occurs at time $\nu=1$.

In the multislot change detection problem, the change occurs in only a subset of the $T$ time slots in each period. To capture this we now introduce a new notation for the post-change law emphasizing the slots where the density changes.
For a subset $S \subset \{1,2,\dots, T\}$, define a possible post-change i.p.i.d. law as
\begin{equation}\label{eq:gSslotdensity}
	\begin{split}
		g_S &= (g_{S,1}, g_{S,2}, \dots, g_{S,T}),
	\end{split}
\end{equation}
with
\begin{equation}\label{eq:gSslotdensity_1}
	g_{S,i} = \begin{cases}
		g_i, \text{ if } \; i \in S \\
		f_i, \text{ if } \; i \not \in S.
	\end{cases}
\end{equation}
Note that
$$
g_S  = (g_1, \cdots, g_T), \; \text{ if } S = \{1,\dots, T\}.
$$
Thus, the set $S$ denotes the slots in which the change occurs:
\begin{equation}\label{eq:CPmultislot}
	X_n \sim
	\begin{cases}
		f_n, &\quad \forall n < \nu \\
		g_{S,n} &\quad \forall n \geq \nu,
	\end{cases}
\end{equation}
with the understanding that $g_{S,n+T} = g_{S,n}$, $\forall n$. This set is not known to the decision maker. However, it is known that
$$
S \in \mathcal{S} \subset 2^{\{1,\dots, T\}},
$$
i.e., $S$ belongs to a family of the subsets of the power set of $\{1,\dots,T\}$. For example,
$$
\mathcal{S} = \{S: |S| \leq m\},
$$
where $|S|$ denotes the size of set $S$. The algorithms that we will propose for change detection will be especially useful when $m \ll T$.

Let $\tau$ be a stopping time for the process $\{X_n\}$, i.e., a positive integer-valued random variable such that the event $\{\tau \leq n\}$ belongs
to the $\sigma$-algebra generated by $\{X_1, \cdots, X_n\}$. We model the change point $\nu$ as a random variable with a prior $\pi$:
$$
\pi_n = \Prob(\nu =n), \quad \text{ for } n = 1, 2, \cdots.
$$
Let $\Prob^S_n$ denotes the law of the observation process $\{X_n\}$ when the change occurs in slots $S$ at time $\nu=n$ and define
$$
\Prob^S_\pi = \sum_{n=1}^\infty \pi_n \; \Prob^S_n.
$$
We use $\Expect^S_n$ and $\Expect^S_\pi$ to denote the corresponding expectations.
For each $S \in \mathcal{S}$, we seek a solution to
\begin{equation}\label{eq:ShiryProbmultislot}
	\min_{\tau \in \mathbf{C}_\alpha} \Expect^S_\pi\left[ \tau - \nu | \tau \geq \nu \right],
\end{equation}
where
\begin{equation}\label{eq:Calpha}
	\mathbf{C}_\alpha = \{\tau: \Prob^S_\pi(\tau < \nu) \leq \alpha\},
\end{equation}
and $\alpha$ is a given constraint on the probability of a false alarm. In fact, we seek an algorithm that solves the above problem uniformly over every $S$.

\subsection{Algorithm for Multislot Change Detection}
Consider a mixing distribution or a probability mass function on the set $\mathcal{S}$:
$$
p_S \geq 0, \; \forall S \in \mathcal{S}, \quad \text{ and } \quad \sum_{S \in \mathcal{S}} p_S = 1.
$$
Define the mixture statistic
\begin{equation}
	\label{eq:ShirMixstat}
	R_n = \frac{1}{\Pi_n} \sum_{k=1}^n \pi_k \sum_{S \in \mathcal{S}} p_S \prod_{i=k}^n \frac{g_{S,i}(X_i)}{f_i(X_i)},
\end{equation}
where $\Pi_n = \Prob(\nu > n)$, and the stopping rule
\begin{equation}
	\label{eq:Shirmixturestop}
	\tau_{mps} = \inf\{n \geq 1: R_n > A\}.
\end{equation}
In the following, we call this algorithm the mixture periodic Shiryaev or the MPS algorithm. Note that
$$
R_n = \sum_{S \in \mathcal{S}} p_S \frac{1}{\Pi_n} \sum_{k=1}^n \pi_k \prod_{i=k}^n \frac{g_{S,i}(X_i)}{f_i(X_i)}.
$$
Thus, the statistic $R_n$ is a mixture of $|\mathcal{S}|$ periodic Shiryaev statistics (see \cite{bane-tit-2021} and Section~\ref{sec:TITreview}), one for each $S \in \mathcal{S}$. Thus, the statistic $R_n$ can be computed recursively and using finite memory for geometric prior $\pi = \text{Geom}(\rho)$ (see Lemma 5.1 in \cite{bane-tit-2021}).

\subsection{Lower Bound on Detection Delay}
In this section, we obtain a lower bound on the average detection delay for any stopping time that satisfies the constraint on the probability of false alarm \eqref{eq:Calpha}. We make the following assumptions.

\begin{itemize}
	\item[(A1)] Let there exist $d\geq 0$ such that
	\begin{equation}\label{eq:priortail_multi}
		\lim_{n \to \infty} \frac{\log \Prob(\nu > n)}{n} = -d.
	\end{equation}
	\item[(A2)] Also, let
	\begin{equation}\label{eq:priortaillogpi}
		\sum_{n = 1}^\infty \pi_n |\log \pi_n| < \infty.
	\end{equation}
\end{itemize}

If $\pi = \text{Geom}(\rho)$, then
$$
\frac{\log \Prob(\nu > n)}{n}  = \frac{\log (1-\rho)^n}{n} = \frac{n \log (1-\rho)}{n} = \log(1-\rho).
$$
Thus, $d = |\log(1-\rho)|$. In addition,
$$
\sum_{n = 1}^\infty \pi_n |\log \pi_n| = \frac{1-\rho}{\rho} \log \frac{1}{(1-\rho)} + \log \frac{1}{\rho}< \infty.
$$
Thus, conditions (A1) and (A2) above are satisfied by the geometric prior.

We first start with a lemma whose proof is elementary. Define
\begin{equation}\label{eq:ZiLL}
	Z_i = \log \frac{g_{S,i}(X_i)}{f_i(X_i)},
\end{equation}
and
\begin{equation}\label{eq:KLnumberS}
	I_S = \frac{1}{T}\sum_{i\in S} D(g_{S,i} \; \| \; f_i).
\end{equation}

\medspace
\begin{lemma}\label{lem:AlmostsureIS}
	For $Z_i$ defined in \eqref{eq:ZiLL} and $I_S$ defined in \eqref{eq:KLnumberS}, for each $S \in \mathcal{S}$ and as $n \to \infty$, we have
	\begin{equation}
		\frac{1}{n} \sum_{i=k}^{k+n-1} Z_i  \; \; \to \; \; I_S, \quad \Prob^S_k \; \text{ a.s.}, \; \forall k \geq 1.
	\end{equation}
\end{lemma}
The lower bound is supplied by the following theorem.

\medspace
\medspace
\begin{theorem}\label{thm:LB_multislot}
	Let the information number $I_S$ be as defined in \eqref{eq:KLnumberS}.  Also, let the prior $\pi$ satisfy the condition (A1) in \eqref{eq:priortail_multi}.
	Then, for any stopping time $\tau \in \mathbf{C}_\alpha$, we
	have
	\begin{equation}
		\Expect^S_\pi\left[ \tau - \nu | \tau \geq \nu \right] \geq \frac{|\log \alpha| }{I_S + d}(1+o(1)), \quad \text{ as } \alpha \to 0.
	\end{equation}
	Here $o(1) \to 0$ as $\alpha \to 0$.
\end{theorem}
\begin{proof}
	The result follows from Lemma~\ref{lem:AlmostsureIS} above and Theorem 5.1 in \cite{bane-tit-2021} (see also Section~\ref{sec:TITreview}).
\end{proof}

\medspace
\medspace

\subsection{Optimality of the MPS Algorithm}
We now show that the MPS algorithm \eqref{eq:Shirmixturestop} is asymptotically optimal for problem \eqref{eq:ShiryProbmultislot} for each post-change slots $S \in \mathcal{S}$, as the false
alarm constraint $\alpha \to 0$. We first prove an important lemma.

Define
\begin{equation}\label{eq:gammakcomplete}
	\gamma_k(\epsilon) = \sum_{n=1}^{\infty} \Prob^S_k \left(\bigg| \frac{1}{n}  \sum_{i=k}^{k+n-1} Z_i - I_S \bigg| > \epsilon\right),
\end{equation}
where $Z_i$ is as defined in \eqref{eq:ZiLL}.

\medspace
\begin{lemma}\label{lem:completeconv}
	For every $\epsilon > 0$,
	\begin{equation}
		\sum_{k=1}^\infty \pi_k \gamma_k(\epsilon) < \infty,
	\end{equation}
	where $\gamma_k(\epsilon) $ is defined in \eqref{eq:gammakcomplete}.
\end{lemma}
\begin{proof}
	The sequence
	$$
	\left\{\gamma_k(\epsilon) \right\}_{k=1}^\infty
	$$
	in \eqref{eq:gammakcomplete}
	is periodic with period $T$, and as a result, there are at most $T$ distinct values in the above sequence:
	$$
	\gamma_1(\epsilon), \; \cdots, \gamma_T(\epsilon).
	$$
	Thus, if we show that these $T$ numbers are finite for any $\epsilon>0$, then we automatically have
	$$
	\sum_{k=1}^\infty \pi_k \gamma_k(\epsilon) \leq \max_{1 \leq k \leq T} \gamma_k(\epsilon) < \infty, \quad \forall \epsilon > 0.
	$$
	Furthermore, since we do not make any explicit assumptions on the actual values taken by the densities $(f_1, \cdots, f_T)$
	and $(g_1, \cdots, g_T)$, we can exploit the i.p.i.d. nature of the processes to just show that
	$$
	\gamma_1(\epsilon) < \infty, \quad \forall \epsilon > 0.
	$$
	
	Recall the definition of $\gamma_1(\epsilon) $:
	\begin{equation*}
		\gamma_1(\epsilon) = \sum_{n=1}^{\infty} \Prob^S_1 \left(\bigg| \frac{1}{n}  \sum_{i=1}^{n} Z_i - I_S \bigg| > \epsilon\right).
	\end{equation*}
	For $\ell =1, 2, \cdots, T$, define
	$$
	Z_i^{(\ell)} = \begin{cases}
		Z_i \quad \text{ if } i=mT+\ell \; \text{ for } m=0,1,2, \cdots\\
		0 \quad \text{ otherwise }.
	\end{cases}
	$$
	Note that
	$$
	Z_i^{(\ell)} = 0, \; \text{ if } \ell \not \in S.
	$$
	Also recall that
	$$
	I_S = \frac{1}{T}\sum_{\ell \in S} I_\ell,
	$$
	where
	$$
	I_\ell = D(g_\ell \; \| \; f_\ell).
	$$
	
	Using these definitions, we can write
	\begin{equation*}
		\begin{split}
			\frac{1}{n} \sum_{i=1}^n Z_i - I_S &= \frac{1}{n} \sum_{i=1}^n \sum_{\ell \in S } Z^{(\ell)}_i - \frac{1}{T}\sum_{\ell \in S} I_\ell 
			=\sum_{\ell \in S} \left(\frac{1}{n} \sum_{i=1}^n Z_i^{(\ell)} - \frac{I_\ell}{T}\right).
		\end{split}
	\end{equation*}
	This implies
	\begin{equation*}
		\begin{split}
			\Prob^S_1 & \left(\bigg| \frac{1}{n}  \sum_{i=1}^n Z_i - I_S\bigg| > \epsilon \right) 
			\leq \sum_{\ell \in S}  \Prob^S_1 \left( \bigg| \frac{1}{n} \sum_{i=1}^n Z_i^{(\ell)} - \frac{I_\ell}{T}\bigg| > \frac{\epsilon}{|S|} \right).
		\end{split}
	\end{equation*}
	Thus, to show the summability of the LHS in the above equation, we need to show the summability of each of the $|S|$ terms on the RHS. Again, due
	to the i.p.i.d. nature of the processes, and because we have made no explicit assumptions about the densities $(f_1, \cdots, f_T)$
	and $(g_1, \cdots, g_T)$, it is enough to establish the summability of any one of the terms on the right.
	That is, for $\ell \in S$, we want to show that
	$$
	\sum_{n=1}^\infty \Prob^S_1 \left( \bigg| \frac{1}{n} \sum_{i=1}^n Z_i^{(\ell)} - \frac{I_\ell}{T}\bigg| > \frac{\epsilon}{|S|} \right) < \infty.
	$$
	
	Define for $\ell =1,2, \cdots, T$,
	$$
	I_i^{(\ell)} = \begin{cases}
		I_\ell \quad \text{ if } i=mT+\ell \; \text{ for } m=0,1,2, \cdots, \ell \in S\\
		0 \quad \text{ otherwise.}
	\end{cases}.
	$$
	Using this definition we write for $\ell \in S$,
	$$
	\frac{1}{n} \sum_{i=1}^n Z_i^{(\ell)} - \frac{I_\ell}{T} =  \frac{1}{n} \sum_{i=1}^n (Z_i^{(\ell)} - I_i^{(\ell)} ) + \frac{1}{n} \sum_{i=1}^n I_i^{(\ell)} - \frac{I_\ell}{T}.
	$$
	
	Thus, with $\tilde{Z}_i^{(\ell)}  = Z_i^{(\ell)} - I_i^{(\ell)}$, we have
	\begin{equation}\label{eq:thm2_22}
		\begin{split}
			\Prob^S_1  \left( \bigg| \frac{1}{n} \sum_{i=1}^n Z_i^{(\ell)} - \frac{I_\ell}{T}\bigg| > \frac{\epsilon}{|S|} \right) 
			&\leq \Prob^S_1 \left(\bigg| \frac{1}{n} \sum_{i=1}^n \tilde{Z}_i^{(\ell)}\bigg| > \frac{\epsilon}{2|S|} \right) 
			\\
   &+ \Prob^S_1  \left(\bigg| \frac{1}{n} \sum_{i=1}^n I_i^{(\ell)} - \frac{I_\ell}{T}\bigg| > \frac{\epsilon}{2|S|} \right).
		\end{split}
	\end{equation}
	Now, $
	\frac{1}{n} \sum_{i=1}^n I_i^{(\ell)} = \frac{1}{n} I_\ell \lfloor \frac{n}{T} \rfloor \to \frac{I_\ell}{T}, \quad \text{ as } k \to \infty$.
	Thus, for $n$ large enough, the second term on the right in \eqref{eq:thm2_22} is identically zero. Thus, we only need to show that
	$$
	\sum_{n = 1}^\infty \Prob^S_1 \left(\bigg| \frac{1}{n} \sum_{i=1}^n \tilde{Z}_i^{(\ell)}\bigg| > \frac{\epsilon}{2|S|} \right) < \infty, \quad \ell \in S.
	$$

	Towards this end, note that in the term $\Prob^S_1 \left(\big| \frac{1}{n} \sum_{i=1}^n \tilde{Z}_i^{(\ell)}\big| > \frac{\epsilon}{2|S|} \right)$, the sum
 $\sum_{i=1}^n \tilde{Z}_i^{(\ell)}$ is updated only once in $T$ time steps. As a result, as a function of $n$, the probability decreases monotonically between $kT+\ell$ and $(k+1)T+\ell-1$, for every $k=0,1,2,\dots$.  Using this fact, we can write
	\begin{equation*}
		\begin{split}
			\sum_{n=1}^\infty \Prob^S_1 \left(\bigg| \frac{1}{n} \sum_{i=1}^n \tilde{Z}_i^{(\ell)}\bigg| > \frac{\epsilon}{2|S|} \right) 
			&\leq T + T \sum_{j=1}^\infty \Prob^S_1 \left(\bigg| \frac{1}{jT+\ell} \sum_{i=1}^{jT+\ell} \tilde{Z}_i^{(\ell)}\bigg| > \frac{\epsilon}{2|S|} \right) \\
			&\leq T + T \sum_{j=1}^\infty \Prob^S_1 \left(\bigg| \frac{1}{j} \sum_{i=1}^{jT+\ell} \tilde{Z}_i^{(\ell)}\bigg| > \frac{\epsilon}{2|S|} \right).
		\end{split}
	\end{equation*}
	The rightmost summation is finite because the sum inside is a sum of $j$ i.i.d. random variables
	with the distribution of $Z_\ell$ under $\Prob_1$ \cite{tart-niki-bass-2014}. The summation is finite for i.i.d. random variables with finite variance. See also \cite{tartakovsky2019sequential}.
\end{proof}

\medspace
\medspace

In words, the above lemma states that i.p.i.d. processes satisfy the complete convergence condition \cite{tart-niki-bass-2014}, \cite{tartakovsky2019sequential}.


\medskip
\begin{theorem}\label{thm:UB}
	Let the prior satisfy the conditions (A1) and (A2).
	With $A=\frac{1-\alpha}{\alpha}$ in \eqref{eq:Shirmixturestop}, we have for each $S \in \mathcal{S}$,
	$$
	\Prob^S_\pi(\tau_{mps} < \nu) \leq \alpha
	$$
	and
	\begin{equation}\label{eq:ShirPerf}
		\Expect^S_\pi\left[ \tau_{mps} - \nu | \tau_{mps} \geq \nu \right] \leq \frac{|\log \alpha| }{I_S + d}(1+o(1)), \quad \text{ as } \alpha \to 0.
	\end{equation}
\end{theorem}
\begin{proof}
	The results follow directly from Lemma~\ref{lem:completeconv} and arguments provided in \cite{tartakovsky2019sequential}. But, we provide the proof in detail for completeness.

Recall that the mixture statistic is defined as
\begin{equation}
	\label{eq:ShirMixstat_1}
	R_n = \frac{1}{\Pi_n} \sum_{k=1}^n \pi_k \sum_{S \in \mathcal{S}} p_S \prod_{i=k}^n \frac{g_{S,i}(X_i)}{f_i(X_i)},
\end{equation}
where $\Pi_n = \Prob(\nu > n)$, and the stopping rule is defined as 
\begin{equation}
	\label{eq:Shirmixturestop_1}
	\tau_{mps} = \inf\{n \geq 1: R_n > A\}.
\end{equation}

	We first note that for any discrete integer-valued random variable such as a stopping time $\tau$,
	\begin{equation}
	    \label{eq:Etaodecomp}
     \Expect[\tau] = \sum_{n=0}^\infty \Prob(\tau > n) \leq N + \sum_{n=N}^\infty \Prob(\tau > n),
	\end{equation}
	where $N$ is any positive integer. The theorem follows by carefully choosing the value $N$ above and obtaining an upper bound on $\Prob(\tau > n)$.
	
	For  $0 < \epsilon < I_S + d$, set
	$$
	N = N_\alpha =  1 + \Bigl \lfloor\frac{\log (A_\alpha/\pi_k) }{I_S + d - \epsilon}\Bigr \rfloor,
	$$
	where
	$$
	A_\alpha = \frac{1-\alpha}{\alpha}.
	$$
	Using $N=N_\alpha$ and $\tau=(\tau_{mps}-k)^+$ in \eqref{eq:Etaodecomp} we get
	\begin{equation}\label{thm2_1}
		\begin{split}
			\Expect^S_k[(\tau_{mps}-k)^+] &\leq N_\alpha + \sum_{n\geq N_\alpha} \Prob^S_k(\tau_{mps}> k+n) 
			\leq N_\alpha + \sum_{n\geq N_\alpha} \Prob^S_k(R_{k+n} < A_\alpha) \\
			&= N_\alpha + \sum_{n\geq N_\alpha} \Prob^S_k(\log R_{k+n} < \log A_\alpha).
		\end{split}
	\end{equation}
	Now, 
	$$
	R_{k+n} = \frac{1}{\Pi_{k+n}} \sum_{t=1}^{k+n} \pi_t \sum_{S \in \mathcal{S}} p_S \prod_{i=t}^{k+n} \frac{g_{S,i}(X_i)}{f_i(X_i)}.
	$$
	The MPS statistic is lower bounded by
	$$
	R_{k+n}  \geq  \frac{1}{\Pi_{k+n}}  \pi_k \; p_S \prod_{i=k}^{k+n} \frac{g_{S,i}(X_i)}{f_i(X_i)}.
	$$
 Here $S$ on the right is the true post-change multiplot set. 
	Taking logarithms on both sides we get
	\begin{equation}\label{thm2_2}
		\begin{split}
			\log(R_{k+n}) &\geq  |\log\Pi_{k+n}| + \log(\pi_k) +  \log(p_S) + \sum_{i=k}^{k+n} Z_i.
		\end{split}
	\end{equation}
	Using \eqref{thm2_2} we can bound the probability in \eqref{thm2_1} for $n \geq N_\alpha$,
	
	\begin{equation}\label{thm2_3}
		\begin{split}
			\Prob^S_k (\log R_{k+n} < \log A_\alpha) \leq &\; \Prob^S_k  \left(  |\log\Pi_{k+n}| + \log(p_S)  + \sum_{i=k}^{k+n} Z_i < \log (A_\alpha/\pi_k) \right) \\
			= &\; \Prob^S_k  \left( \frac{1}{n}\sum_{i=k}^{k+n} Z_i  + \frac{|\log\Pi_{k+n}|}{n} + \frac{\log( p_S)}{n}  < \frac{\log (A_\alpha/\pi_k)}{n} \right) \\
			\leq &\; \Prob^S_k  \left( \frac{1}{n}\sum_{i=k}^{k+n} Z_i  + \frac{|\log\Pi_{k+n}|}{n} + \frac{\log( p_S)}{n}  < I_S + d -\epsilon \right) \\
			= &\; \Prob^S_k  \left( \frac{1}{n}\sum_{i=k}^{k+n} Z_i   < I_S + d - \frac{|\log\Pi_{k+n}|}{n} - \frac{\log(p_S)}{n} -\epsilon \right).
		\end{split}
	\end{equation}
	Now select $\alpha$ small enough so that for every $n \geq N_\alpha$
	\begin{equation*}
		\begin{split}
			\Big| d - \frac{|\log\Pi_{k+n}|}{n} \Big|&< \frac{\epsilon}{4}, \\
			\Big|\frac{\log(p_S)}{n} \Big| &< \frac{\epsilon}{4}.
		\end{split}
	\end{equation*}
	Specifically, select $\alpha$ small enough such that the statements in the above display are true for all $n$ satisfying
	$$
	n \geq 1 + \Bigl \lfloor\frac{\log (A_\alpha) }{I_S + d - \epsilon}\Bigr \rfloor.
	$$
	This ensures that the chosen small $\alpha$ is not a function of the index $k$.
	This gives us
	\begin{equation}\label{thm2_4}
		\begin{split}
			\Prob^S_k &(\log R_{k+n} < \log A_\alpha) \\
			&\leq \; \Prob^S_k  \left( \frac{1}{n}\sum_{i=k}^{k+n} Z_i   < I_S + d - \frac{|\log\Pi_{k+n}|}{n} - \frac{\log(p_S)}{n} -\epsilon \right)\\
			&\leq \; \Prob^S_k  \left( \frac{1}{n}\sum_{i=k}^{k+n} Z_i   < I_S - \frac{\epsilon}{2} \right).
		\end{split}
	\end{equation}
	Substituting this in \eqref{thm2_1} we get for $\alpha$ small enough, uniformly over $k$,
	\begin{equation}\label{thm2_5}
		\begin{split}
			\Expect^S_k[(\tau_{mps}-k)^+]
			&\leq N_\alpha + \sum_{n\geq N_\alpha} \Prob^S_k(\log R_{k+n} < \log A_\alpha) \\
			&\leq N_\alpha + \sum_{n\geq N_\alpha} \Prob^S_k  \left( \frac{1}{n}\sum_{i=k}^{k+n} Z_i   < I_S - \frac{\epsilon}{2} \right) \\
			&\leq N_\alpha + \sum_{n=1}^\infty \Prob^S_k  \left( \Big| \frac{1}{n}\sum_{i=k}^{k+n} Z_i   - I_S \Big| > \frac{\epsilon}{2} \right) \\
			&=N_\alpha + \gamma_k(\epsilon/2).
		\end{split}
	\end{equation}
	This gives us
	\begin{equation}\label{thm2_6}
		\begin{split}
			\Expect^S_\pi  [(\tau_{mps}-\nu)^+] &=  \sum_{k=1}^\infty \pi_k \Expect^S_k[(\tau_{mps}-k)^+] \\
			&\leq \sum_{k=1}^\infty \pi_k N_\alpha + \sum_{k=1}^\infty \pi_k\gamma_k(\epsilon/2)\\
			&\leq \sum_{k=1}^\infty \pi_k \left(1 +  \frac{\log (A_\alpha/\pi_k) }{I_S + d - \epsilon}\right) + \sum_{k=1}^\infty \pi_k\gamma_k(\epsilon/2)\\
			&=\frac{\log (A_\alpha) }{I_S + d - \epsilon} + \text{ constant}.
		\end{split}
	\end{equation}
	The remaining term is a constant (not a function of $\alpha$) because of the assumptions made in the theorem statement and due to Lemma~\ref{lem:completeconv}.
	Finally,
	\begin{equation}\label{thm2_7}
		\begin{split}
			\Expect^S_\pi [\tau_{mps}-\nu | \tau_{mps} \geq \nu] &= \frac{\Expect^S_\pi[(\tau_{mps}-\nu)^+]}{\Prob^S_\pi(\tau_{mps} \geq \nu)} \\
			&\leq \frac{\frac{\log (A_\alpha) }{I_S + d - \epsilon} + \text{ constant}}{1-\alpha} \\
			&=\frac{|\log \alpha |}{I_S + d - \epsilon} (1+ o(1)), \quad \text{ as } \alpha \to 0.
		\end{split}
	\end{equation}
	The result now follows because $\epsilon$ can be made arbitrarily small.
	
	The fact that setting $A=A_\alpha$ guarantees that the false alarm constraint is satisfied follows from \cite{tartakovsky2019sequential}.
	
\end{proof}
Thus, the MPS algorithm achieves the asymptotic lower bound given in Theorem~\ref{thm:LB_multislot} and is asymptotically optimal uniformly over $S$.

\section{Numerical Results}
	\subsection{Applying the Periodic-CUSUM Algorithm to Los Angeles Traffic Data}
 \label{sec:traffic}
In this section, we demonstrate how to train and apply the periodic-CUSUM algorithm \eqref{eq:PeriodicCUSUM} on traffic flow indicator data of the Los Angeles' (LA) highway. For ease of reference, we reproduce the algorithm here. 
\begin{equation}\label{eq:repeatedPeriodicCUSUM}
	W_{n+1} = W_n^{+} + \log \frac{g_{n+1}^{(1)}(X_{n+1})}{g_{n+1}^{(0)}(X_{n+1})}; \; \; W_n=0.
\end{equation}
\begin{equation}\label{eq:repeatedPeriodicCUSUMstop}
	\tau_c = \inf \{n \geq 1: W_n \geq A\}.
\end{equation}
We train the pre-change model using weekday traffic data and the post-change model using weekend or holiday data. We then apply the algorithm to the traffic data to detect weekend or holiday traffic. We now discuss the application in detail.

We applied the algorithm at selected stations along the Westbound I-10 highway in LA  County (see Fig. \ref{fig:figI10}) by incorporating multiple historical traffic flow datasets that were freely accessible via Performance Measurement System (PeMS website https://pems.dot.ca.gov/) for a time span of August 2020 - September 2021. The traffic counts reported by PeMS act as a proxy for the traffic flow of each ten stations spaced about 0.33 miles apart, as seen in Fig. \ref{fig:figI10}. PeMS data in archive format is released daily.

\par
Downloaded data for one station in LA has multiple fields in a comma-delimited text format. Each field contains useful traffic attributes such as timestamps, station IDs, number of vehicles/5-minute bins, and so on (see Fig. \ref{fig:headerPEMS} for a sample of this data). In addition, we were able to locate each station precisely on the map of Los Angeles by using station$\_$id (second column in Fig. \ref{fig:headerPEMS}) as a key from PeMS and merging those with an additional table that contains geographical coordinates [lat, lon] for selected traffic stations as illustrated in Fig. \ref{fig:figI10}.

\begin{figure}[htbp]
	\centering
	\includegraphics[width=0.7\linewidth,trim={1cm 1cm 1cm 1cm},clip]{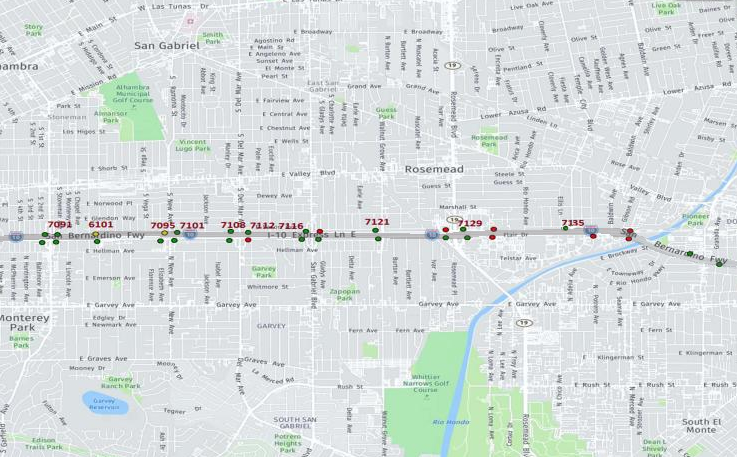}
	\caption[Geographical Location of Selected Station.]{Selected stations map. The station ID is denoted by the numbers for westbound traffic: The main features are depicted, while a light-grey hue highlights a segment of the I-10 highway (Christopher Columbus Transcontinental Highway), LA, CA.}
	\label{fig:figI10}
\end{figure}

\begin{figure}[htbp]
	\centering
	\includegraphics[width=0.7\linewidth]{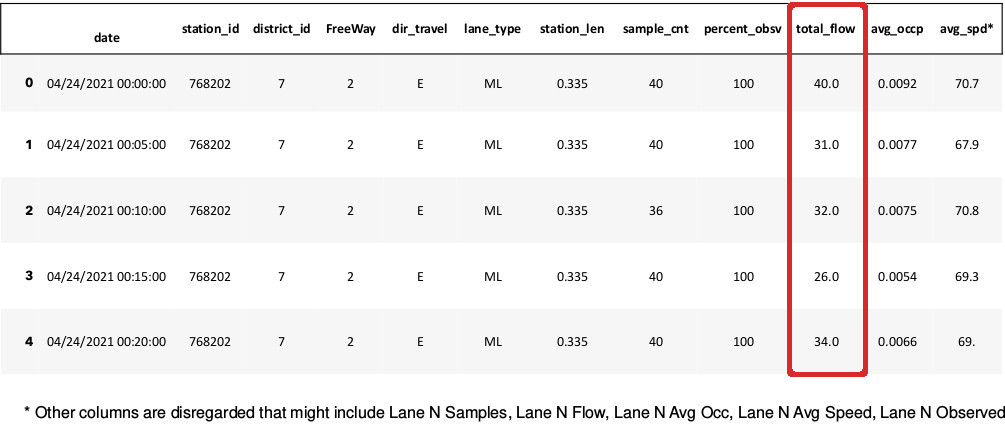}
	\caption{Sample raw data of an eastbound vehicle detector station (VDS) with ``station$\_$id" (column two) ending with 8202. Note that spotted column for ``total$\_$flow" (column ten) is defined as the number of vehicles/five-minute.}
	\label{fig:headerPEMS}
\end{figure}

\par

\par
We observed similar patterns on different days of the week as well as across multiple stations within the same segment of interest (Fig.~\ref{fig:fig7095}). 
We chose the sensor traffic with an ID ending with 7095 on a random weekday 
of the month of August 2021 for training purposes, and we left out the last month of September 2021 for the test set. The comparison of sample paths of August's Mondays with Labor Day Holidays of 2020 (label 249) and 2021 from the test using the smoothing technique is depicted in Fig.~\ref{fig:data249}. Because the readings were noisy, we used the station's median moving average (MMA) of the previous hour's samples (we dropped the duplicates while keeping the last in our station dataset for future analysis).

\begin{figure}[htbp]
    \centering
    \subfloat[\centering\label{sub:fig7095a}]{{\includegraphics[width=.4\linewidth]{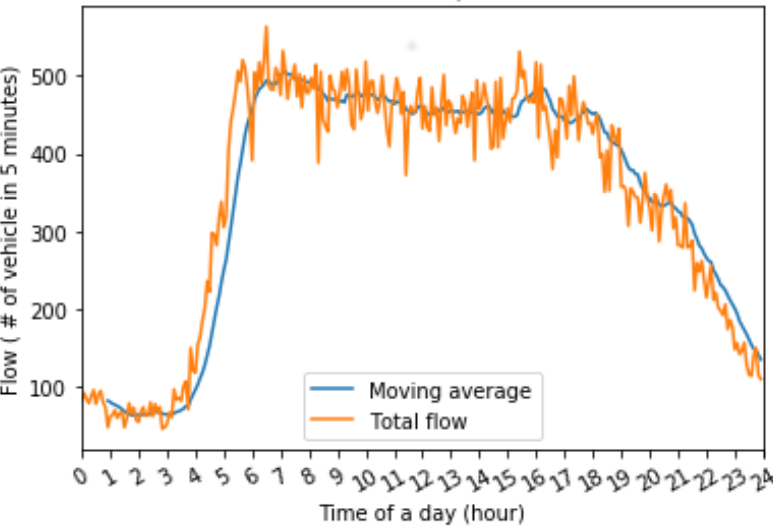}}}
    \qquad
    \subfloat[\centering \label{sub:fig7095b}]{{\includegraphics[width=.4\linewidth]{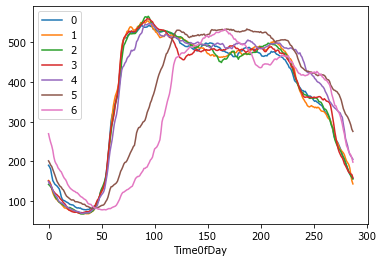}}}
    \caption{(a): Illustration of sample path from the station with index three (7095) over a day (288 bins). (b): Average traffic counts on different days of Aug (Mon=0), 2021.}
    \label{fig:fig7095}
\end{figure}

\begin{figure}[htbp]
    \centering
    \subfloat[\centering \label{sub:anomalyDay}]{{\includegraphics[width=0.45\linewidth]{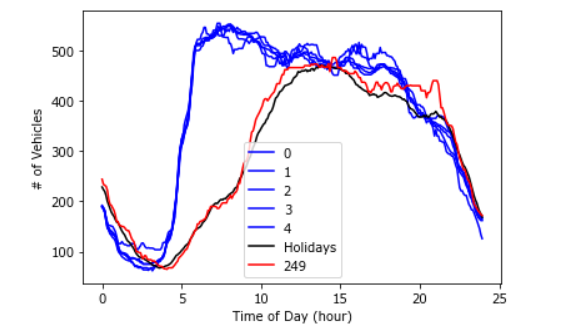}}}
     \qquad
     \subfloat[\centering \label{sub:figCUSUMLike}]{{ \includegraphics[width=0.45\linewidth]{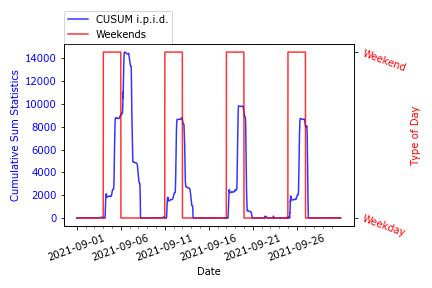}}}
\caption{(a) Comparisons between some of the training sample paths in normal (August's Mondays) vs Labor Day of 2021 and Labor day of 2020 (Label 249). (b): Labor week periodic CUSUM test statistics and event labels for 09/06/2021.}
\label{fig:data249}
\end{figure}
\par
For applying the periodic CUSUM algorithm, we assumed that $T=288 \; (12\times 24)$ (number of bins), and modeled 
\begin{equation}
    \begin{split}
        g_{i}^{(0)} &= \text{Poisson}(\lambda_i^{(0)}), \quad i=1,2, \dots, T\\
        g_{i}^{(1)} &= \text{Poisson}(\lambda_i^{(1)}), \quad i=1,2, \dots, T.
    \end{split}
\end{equation}
We then learned the Poisson parameters from the training data. In Fig.~\ref{sub:figCUSUMLike}, we have plotted the periodic CUSUM statistic for the test data. The red rectangular blocks indicate the location of weekends and the blue curve is the test statistic. As seen in the figure, the test statistic rose sharply around the weekends to indicate that a change in the traffic flow has been detected. 


\subsection{Numerical result for Multislot Quickest Change Detection}\label{sec:secMultit}

In this section, we apply the MPS algorithm \eqref{eq:ShirMixstat} on simulated noisy sinusoidal data. For ease of reference, we reproduce the algorithm here. 
\begin{equation}
	\label{eq:rePeatedShirMixstat}
	R_n = \frac{1}{\Pi_n} \sum_{k=1}^n \pi_k \sum_{S \in \mathcal{S}} p_S \prod_{i=k}^n \frac{g_{S,i}(X_i)}{f_i(X_i)},
\end{equation}
where $\Pi_n = \Prob(\nu > n)$, and the stopping rule
\begin{equation}
	\label{eq:Shirmixturestoprepeat}
	\tau_{mps} = \inf\{n \geq 1: R_n > A\}.
\end{equation}
Specifically, we assume that we observe a noisy version of a sequence of sinusoidal waveforms as shown in Fig.~\ref{sub:absSinRealXaxis}. At the change point, the shape of the sinusoidal signal is distorted as shown in Fig.~\ref{sub:slotchange}. The goal is to detect this distortion in real-time. We assume that we know the type of distortion but don't know the precise location of the distortion. Thus, the actual distortion can be any one of the five shown in Fig.~\ref{sub:allPostChangeParMean4laws}. We will assume that each of the distortions is equally likely for the design of the MPS algorithm \eqref{eq:rePeatedShirMixstat}. 

\begin{figure}
\centering
    \subfloat[\centering \label{sub:absSinRealXaxis}]{{\includegraphics[width=0.45\linewidth]{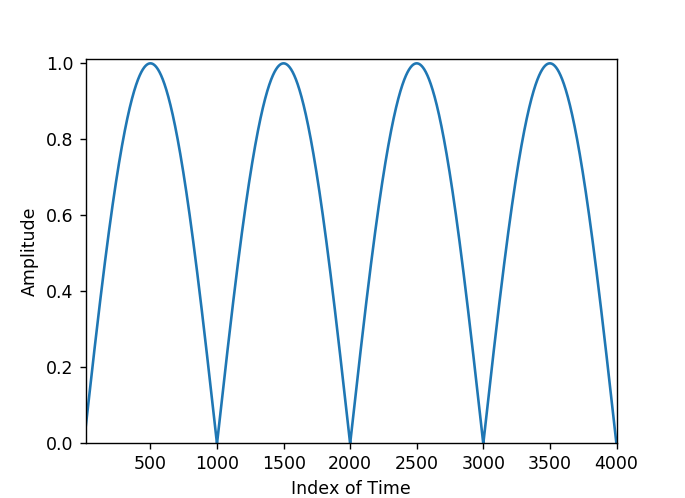}}}%
    \qquad
    \subfloat[\centering\label{sub:slotchange}]{{\includegraphics[width=0.45\linewidth]{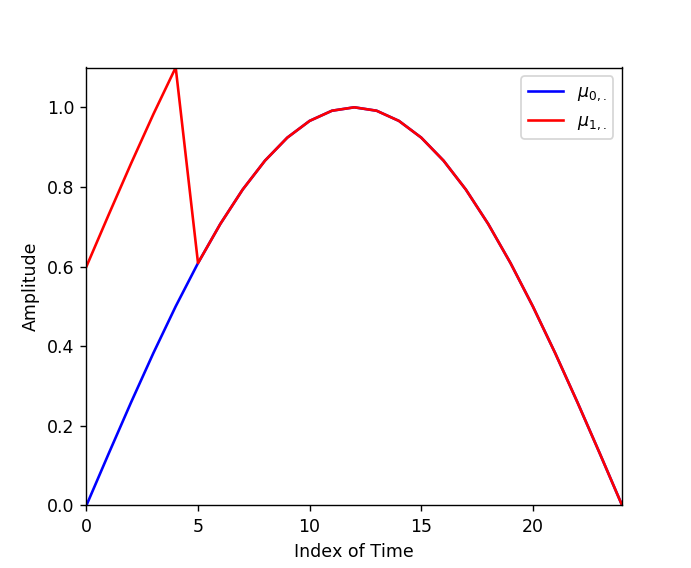}}}%
\caption{(a): Depiction of four positive half sinusoidal waves sampled at a frequency of 1k Hz (Samples/Cycle). (b): Illustration of parameters that governed the PDFs of pre-change and actual post-change that shifted up at the first time slot ([0, 4]).}
\label{fig:preChangeActualPostSlot}
\end{figure}

Let $h(t)$ be the blue sinusoidal signal shown in Fig.~\ref{sub:slotchange}. 
Then, we assume that $T=25$ and create the simulated data using
$$
f_i = \mathcal{N}\left(\mu_{0,i}, 0.01 \right), \quad i=1,2, \dots, 25,
$$
where $\mu_0$ is a $25$-length vector given by 
$$
\mu_{0,i} = h(i), \quad i=1,2, \dots, 25. 
$$
The post-change data is generated from
$$
g_i = \mathcal{N}\left(\mu_{1,i}, 0.01 \right), \quad i=1,2, \dots, 25,
$$
where
\[\mu_{1,i}=\begin{cases}
	\mu_{0,i}, & \text{for i}\notin [0,4]:=[0,1,2,3,4]\\
	\mu_{0,i}+0.6, & \text{for i} \in [0,4].
\end{cases}
\]
Thus, the post-change data is generated by assuming that the true post-change slots are 
$$
S= [0, 4].
$$
But, we assume that the multislot family $\mathcal{S}$ is
$$
\mathcal{S} = \{[0,4], [5, 9], [10, 14], [15, 19], [20, 24]\}
$$
with
$$
p_S = \frac{1}{5}, \quad \text{for all } S \in \mathcal{S}. 
$$
\begin{figure}
	\centering
	\subfloat[\centering \label{sub:allPostChangeParMean4laws}]{{\includegraphics[width=0.5\linewidth]{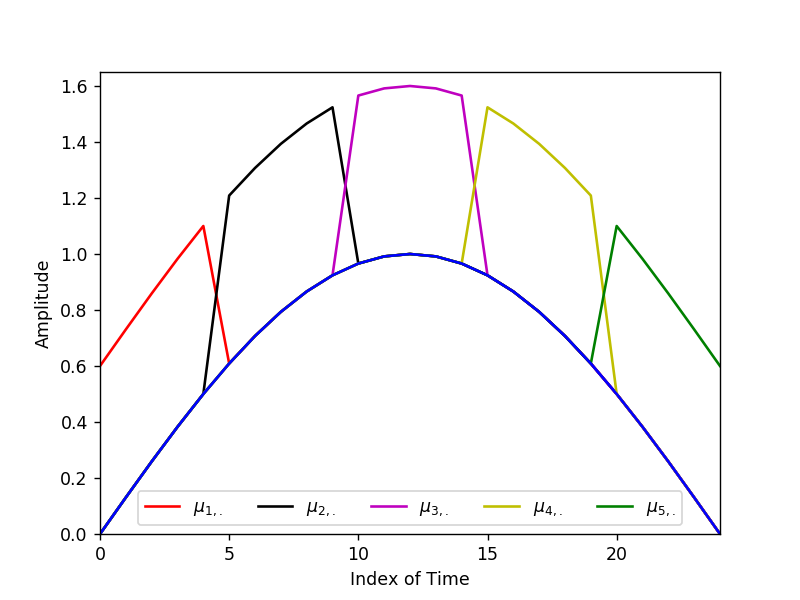}}}
		\subfloat[\centering \label{sub:after4cyclesPostChange125}]{{\includegraphics[width=0.55\linewidth]{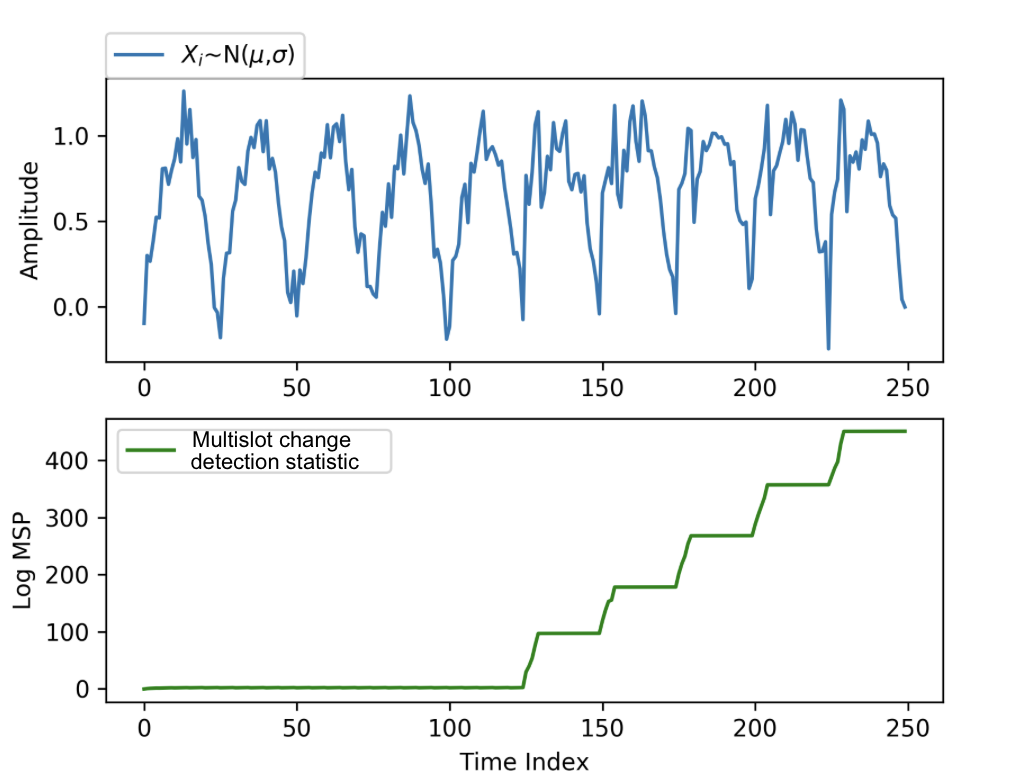}}}
\caption{(a): Depiction of all possible post-change waveforms. (b): Test statistics and sample path for pre/post-change distributions of Gaussian with change-point at time index 125 (at the end of the fifth cycle).}
\end{figure}
We also assumed a geometric prior on the change point with parameter $\rho=0.01$, i.e., $\pi_k = (1-\rho)^{k-1}\rho$. We plot the generated data and the MPS algorithm statistic in Fig.~\ref{sub:after4cyclesPostChange125}. As can be seen from the figure, the algorithm detects the change quite effectively. We repeated the simulation with different change slots. Regardless of the slot index, the MPS algorithm was able to detect the changes with no false alarms.
\par

\subsection{Numerical Results for Robust Quickest Change Detection Algorithm}\label{sec:secRobust}
\begin{figure}
    \centering
    \subfloat[\centering \label{sub:squareSignalPeriodOneIdx3}]{{
    \includegraphics[width=0.5\linewidth]{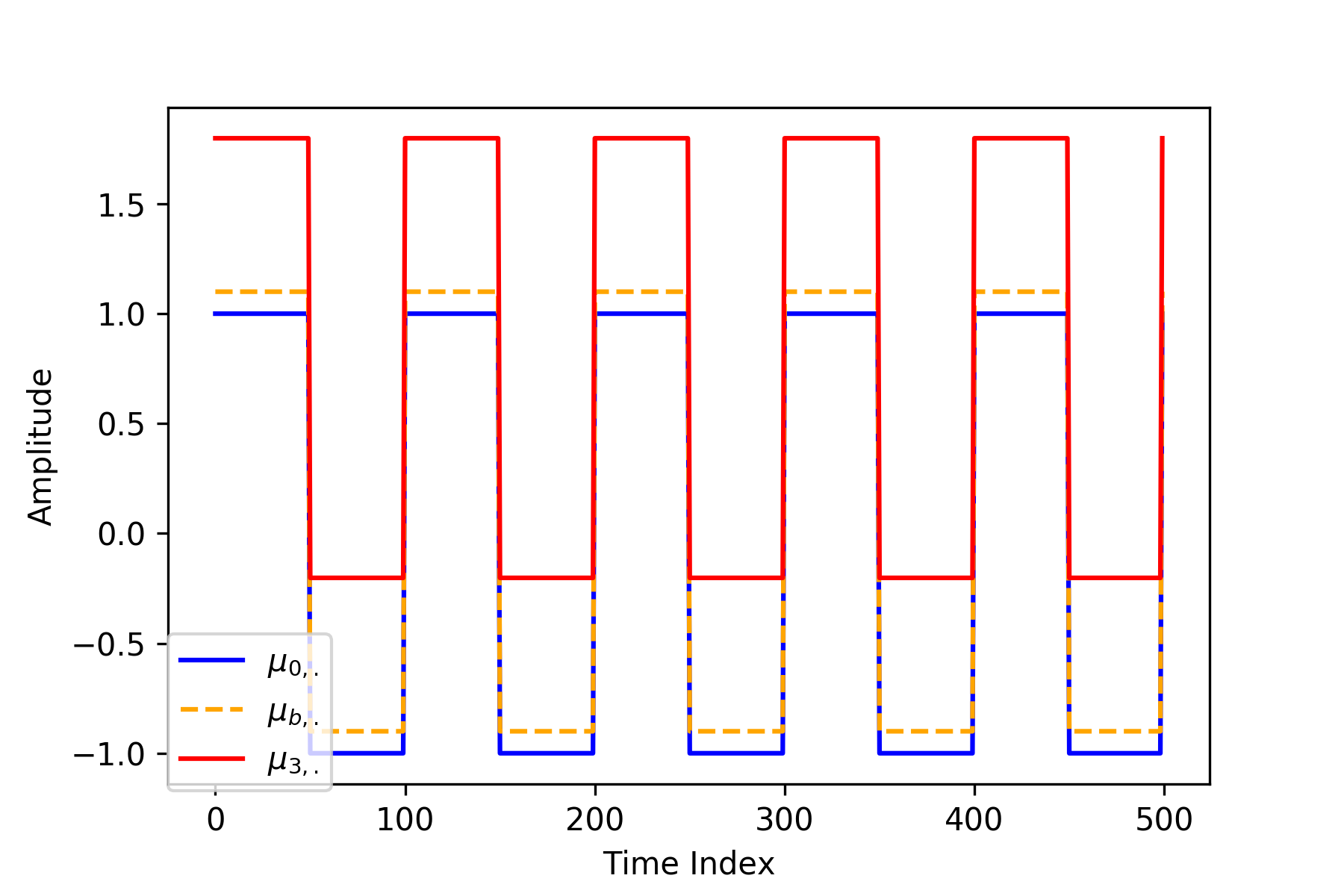}}}%
    \subfloat[\centering \label{sub:LRST2Seq5SquareSignal}]{{
    \includegraphics[width=0.55\linewidth]{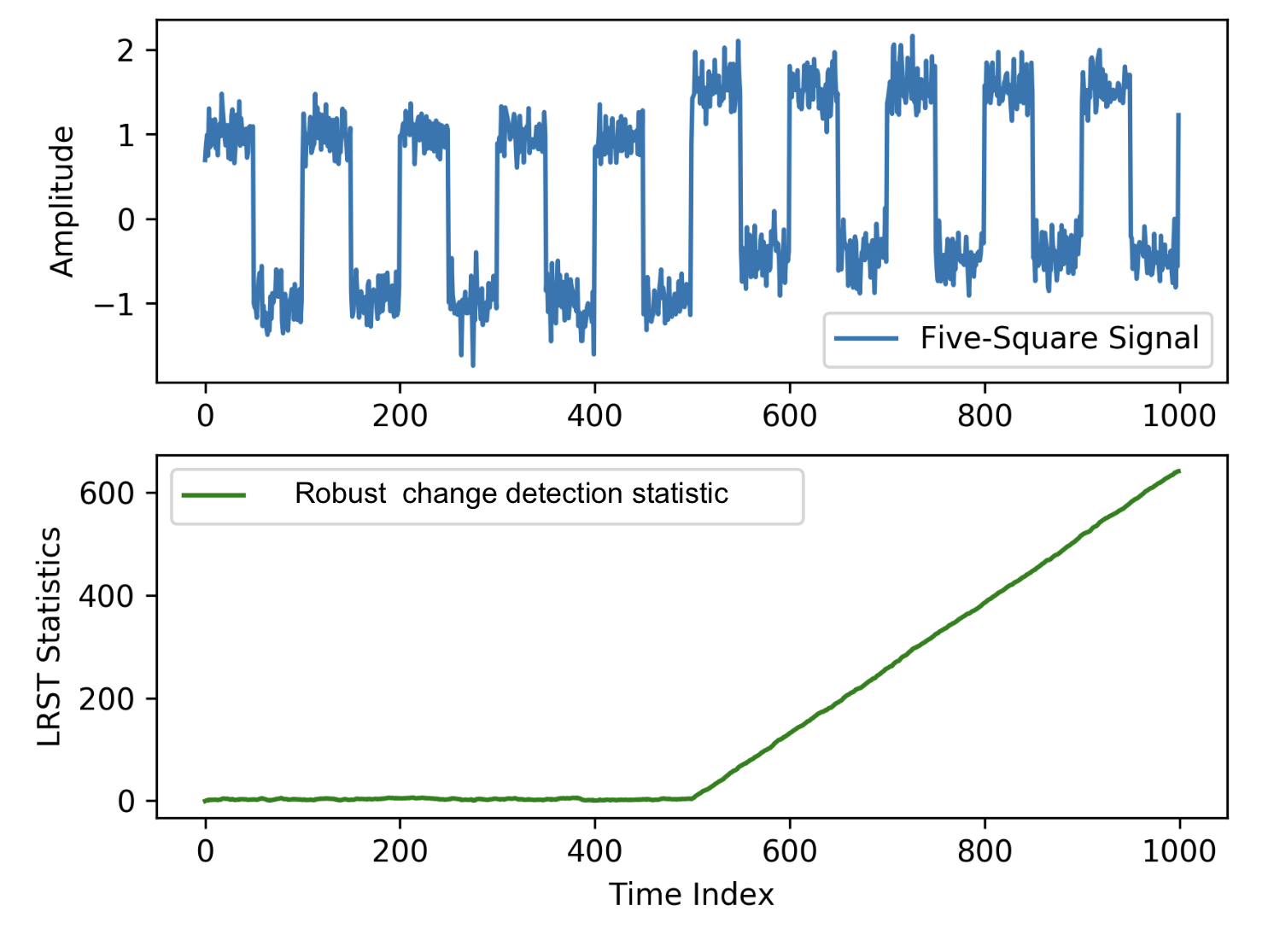}}}%
    \caption{(a): Illustration of different rectangular waveforms involved in robust detection of change in distribution. (b): depictions of a sample path of a five-square signal with shift-up of 0.8 in mean at change-point $\nu$=500 and calculated robust test statistic.}
\end{figure}
In this section, we apply the robust algorithm defined in \eqref{eq:LFLshir} to simulated data. For ease of reference, we reproduce the algorithm here. 
\begin{equation}\label{eq:repeatedLFLshir}
	\bar{\tau}^* = \inf \{n \geq 1: \bar{p}_n \geq A_{(n \bmod T)}\},
\end{equation}
where $\bar{p}_0=0$, and
\begin{equation}\label{eq:repeatedLFPbelief}
	\bar{p}_n = \frac{\tilde{p}_{n-1} \; \bar{g}_n(X_n)}{\tilde{p}_{n-1} \; \bar{g}_n(X_n) + (1-\tilde{p}_{n-1}) f_n(X_n)},
\end{equation}
with
$$
\tilde{p}_{n-1} = \bar{p}_{n-1} + (1-\bar{p}_{n-1}) \rho.
$$
Recall that here $(\bar{g}_1, \dots, \bar{g}_T)$ is the least favorable i.p.i.d. law and $(f_1, \dots, f_T)$ is the pre-change i.p.i.d. law. 

In this numerical experiment, we assume that we observe a noisy version of a rectangular waveform; see Fig.~\ref{sub:squareSignalPeriodOneIdx3}. Before a change point of $500$, the rectangular waveform alternates between $+1$ and $-1$ for $50$ time slots each (blue curve in Fig.~\ref{sub:squareSignalPeriodOneIdx3}). After the change point, the waveform switches between $+1.8$ and $-0.2$ (red waveform in Fig.~\ref{sub:squareSignalPeriodOneIdx3}). We assume that the decision maker is unaware of the exact post-change waveform. But, he/she knows that the deviation will be at least by $0.1$ (dashed orange waveform in Fig.~\ref{sub:squareSignalPeriodOneIdx3}). We assume that we observe the waveform after Gaussian zero-mean random variables with variance $0.01$ have been added. In this setup, it can be shown that the Gaussian i.p.i.d. process with an orange mean level is the least favorable. The generated observation sequence and the robust change detection statistic \eqref{eq:repeatedLFPbelief} are shown in Fig.~\ref{sub:LRST2Seq5SquareSignal}. We used $\rho=0.01$ to generate the statistic. As can be seen from the figure, the robust change detection algorithm effectively detects the change in the waveform pattern. 


\par
\par

\subsection{Numerical Results for ECG Arrhythmia Detection and Fault Isolation}
\label{sec:ECG}
The majority of ECG data are periodic in nature due to the electrical activity within one's heart muscle cells (internal dynamics) over the course of one heartbeat (a PQRS cycle as illustrated in Fig. \ref{fig:ecgSinWaveAnnotated}). 
Due to its diagnostics application, there has recently been a large body of papers on developing algorithms to automate the detection and classification of heart arrhythmia from ECG data using machine learning and statistical pattern recognition perspectives \cite{Can2012IEEETransBioEng,Dechazal2004IE3BioEng,Hannun2019Nature}. In this section, we apply the quickest change detection and fault isolation algorithm for i.p.i.d. processes developed in Section~\ref{sec:QCDFaultIsolation} to real ECG data and simulated wavelet data. Again, for ease of reference, we reproduce the algorithm here. 

For $\ell = 1, \dots, M$, define the stopping times
\begin{equation}\label{eq:repeattauell}
	\begin{split}
		\tau_\ell &= \inf\left\{n \geq 1: \max_{1 \leq k \leq n} \; \min_{0 \leq m \leq M, m \neq \ell} \;  \sum_{i=k}^n \log \frac{g_{i}^{(\ell)}(X_{i})}{g_{i}^{(m)}(X_{i})} \geq A \right\}.
	\end{split}
\end{equation}
The stopping time and decision rule for our detection-classification problem is defined as follows:
\begin{equation}
\label{eq:repeattaudc}
	\begin{split}
		\tau_{dc} &= \min_{1 \leq \ell \leq M} \; \tau_\ell, \\
		\delta_{dc} &= \arg \min_{1 \leq \ell \leq M} \tau_\ell.
	\end{split}
\end{equation}
A window-limited version of the above algorithm is obtained by replacing each $\tau_\ell$ in \eqref{eq:tauell} by
\begin{equation}\label{eq:repeatwindowlimited}
	\begin{split}
		\tilde{\tau}_\ell &= \inf\left\{n: \max_{n-L_\beta \leq k \leq n} \; \min_{0 \leq m \leq M, m \neq \ell} \; \sum_{i=k}^n \log \frac{g_{i}^{(\ell)}(X_{i})}{g_{i}^{(m)}(X_{i})} \geq A \right\}
	\end{split}
\end{equation}
for an appropriate choice of window $L_\beta$. Recall that here 
$(g_{1}^{(0)}, \dots, g_{T}^{(0)}) $ is the normal i.p.i.d. law and 
$(g_{1}^{(\ell)}, \dots, g_{T}^{(\ell)}) $, for $\ell \neq 0$, is the post-change i.p.i.d. law representing anomaly or change of type $\ell$.



\begin{figure}[htbp]
	\centering
    \vspace{1cm}
	\includegraphics[width=1.5cm,height=3cm,trim={5cm 3cm 5cm 3cm}]{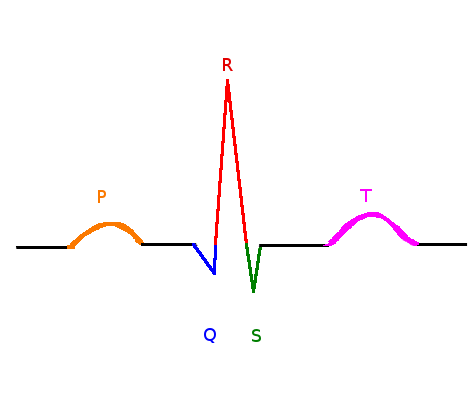}
	\caption{Depiction of morphological features of a normal heartbeat: peaks have been identified.}
	\label{fig:ecgSinWaveAnnotated}
\end{figure}

\subsubsection{MIT-BIH Dataset}
This paper uses the MIT-Boston's Beth Israel University-Hospital (MIT-BIH) dataset downloaded from the Research Resource for Complex Physiologic Signals (PhysioNet) website. The acquired dataset contained 48 recordings from 47 human subjects in which each human subject's data were recorded for about half an hour  \cite{Moody2005}. The data contained information in the form of a 2D array for two-channel signals, a 1D array of expert annotations for the type of arrhythmia, and a 1D array for the location of R-peaks to provide sufficient information for the interpretation of each ECG data. The 2D array signal consists of two-channel sinusoidal waves with an 11-bit resolution over ten milli-volts (mV) range sourced from a 12-lead standard ECG device with a constant sampling rate of 360 samples/second (Hz) for all ECGs in MIT-BIH database \cite{Moody2001}.\\

Since the main-lead II (mlII) channel was the common ECG recording for all patients, the  annotations are only provided for this lead from 12 leads. Thus, we analyzed this array similarly to \cite{Can2012IEEETransBioEng}.\\

\begin{table}
	\centering
	\caption{The equivalent classes of AAMI for normal and abnormal heartbeats in the MIT-BIH dataset.}
	\vspace{.5cm}
	\label{tab:tabMIT2AAMI}
	\begin{tabular}{cccc}
		\hline
		{id}& {Class Rep.}&{Symbol}\\
		\hline
		0& N& `N',`e',`j',`L',`R'\\
		1& V&`V',`E'\\
		2& S& `S', `A',`a',`J'\\
		3& F & `F'\\
		\hline
	\end{tabular}
	\hfill
	\vfill
\end{table}

We used a four-class representation from  the Association for the Advancement of Medical Instrumentation (AAMI) standard to re-cluster different annotations of MIT-BIH into smaller clusters. The standard, which has four larger classes, namely `N' (i.e., any `N,' `e,' `j,' `L,' or `R' from MIT-BIH for normal heartbeat), `S' (supraventricular ectopic beat), `V' (ventricular ectopic beat), and `F' (fusion beat) has been used to re-cluster 12 classes of observed annotations in MIT-BIH into four verified classes in Tab. \ref{tab:tabMIT2AAMI}. As a result of the existence of this table, we grouped each label into a  representative class of AAMI.
\par
We chose the subject patient with the identification (ID) number 208 for a patient-specific analysis. As seen in Tab. \ref{tab:tabStatMITBIH}, the corresponding size of the annotation array for this subject is around 3,000 out of all 112,000 in the MIT-BIH dataset (including the patient with ID=208). As Tab. \ref{tab:tabStatMITBIH} suggests, we removed the supraventricular arrhythmia (cluster of 'S') from the ECG wave. \\


\subsubsection{Data Centering and Standardization}
As illustrated in Fig. \ref{fig:ecgSinWaveAnnotated}, one way of segmentation of data of mlII of ECGs is to obtain an index of mid R-R from heartbeats. Applying partitioning above for patient with ID=208 resulted in 2,951 heartbeats data containing R annotations for main-lead II data with an equal number of annotations excluding `S' and `Q' waves, which ruled 88 annotated heartbeats out of all human annotated heartbeats in the result. Because of the sampling rate of the ECG device, the heartbeats were bounded above by a length of 360.\
All obtained heartbeats had different time lengths. The re-sampling function based on Fast Fourier Transformation (FFT), applied to the length of heartbeats were less than 360 (see an example of this implementation in the lower plot of Fig. \ref{fig:Firtst10HeartBeatsResample}).

\subsubsection{Training and Test Splits}
We randomly sampled 50\% of all heartbeats in three clusters of `N', `V', and `F' of AAMI standard, which accounted for about 48\% of heartbeats for training purposes (the $\#$ of each heartbeat is calculated as Tab. \ref{tab:tabStatTrainingTest}). 
\begin{table}
	\centering
	\caption{Number of re-clustered annotations for the patient with ID=208 vs all labels in MIT-BIH database.}
	\label{tab:tabStatMITBIH}
	\vspace{.5cm}
	\begin{tabular}{ccccccc}
		\hline
		& {`N'} & {`V'}& {`S'}& {`F'} & {`Q'$^*$}& {Total}\\
		\hline
		Patient with ID=208& 1,585 & 992&  2& 373&  86&	3,039\\
		MIT-BIH&  90,631 &  7,236&  2,781&  803&  11,196&	112,647\\
		\hline
	\end{tabular}
	\hfill
	\vfill
	\footnotesize{$^*$ Represents the cluster for all annotations that are not included in the first four clusters.}
\end{table}
For showing the effectiveness of our algorithm, we used a sequence made with ten heartbeats. Because the `S' and `Q' heartbeats might present in any order for the testing real-time situation we only provide testing from the original ECG that all annotations were contained in one of three studied classes and typically will be shown in batches of ten heartbeats for visualization similar to Fig. \ref{fig:Firtst10HeartBeatsResample}.

\begin{table}
	\centering
	\caption{Number of different heartbeats for each cluster in training and test sets chosen from the patient with ID=208.}
	\label{tab:tabStatTrainingTest}
	\vspace{.5cm}
	\begin{tabular}{cccccc}
		\hline
		& {`N'} & {`V'}& {`F'}  & {Total}\\
		\hline
		Training& 817 & 488&  171& 	1,476\\
		Test&  789 &  477&  177& 1,443\\
		\hline
	\end{tabular}
	\hfill
	\vfill
\end{table}

\begin{figure}
	
	\centering
	
	\includegraphics[width=0.7\linewidth]{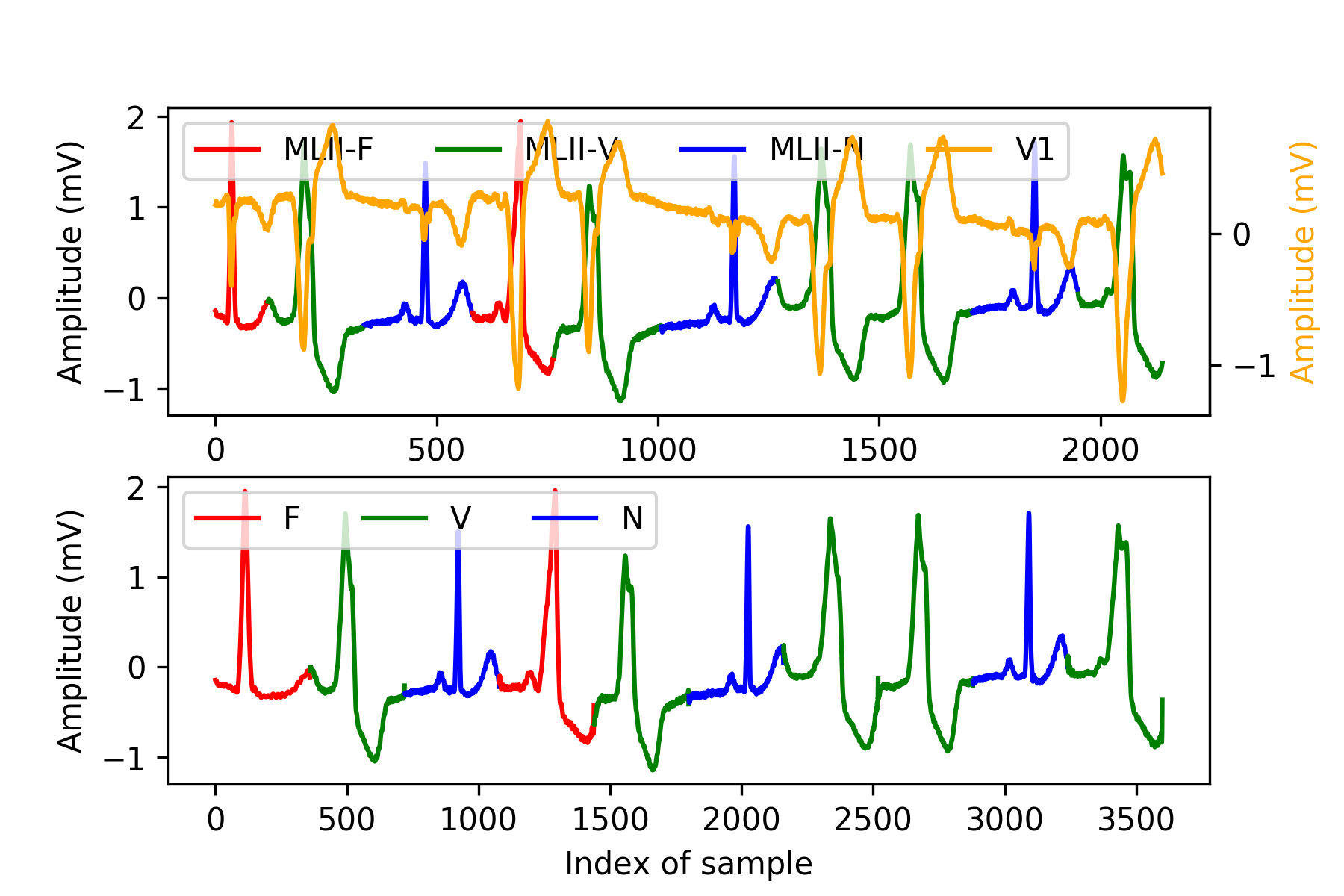}
	\caption{Illustration of the first ten raw heartbeats vs resampled partitioned of MLII data from a two-channel ECG of the patient with ID=208.}
	\label{fig:Firtst10HeartBeatsResample}
	
\end{figure}
To train the i.p.i.d. models for each class, we assumed that the ECG waveforms are deterministic waveforms corrupted by Gaussian noise. We used the training data to learn the means and variances for the Gaussian i.p.i.d. processes. We used $T=360$, the time obtained after resampling. 
The learned mean and variance parameters are shown in Fig. \ref{fig:parPrePostChangeBeatsResample}. The bold lines are the expected values and the dashed lines are one standard deviation away from the mean line. As shown in Fig. \ref{fig:parPrePostChangeBeatsResample}, there are only a few time slots in which two distributions can be separated. Thus, we only focused on discrete time intervals of [130, 155] or [200, 220] to improve the accuracy of predictions. 
\begin{figure}
	\centering
	\includegraphics[width=0.7\linewidth]{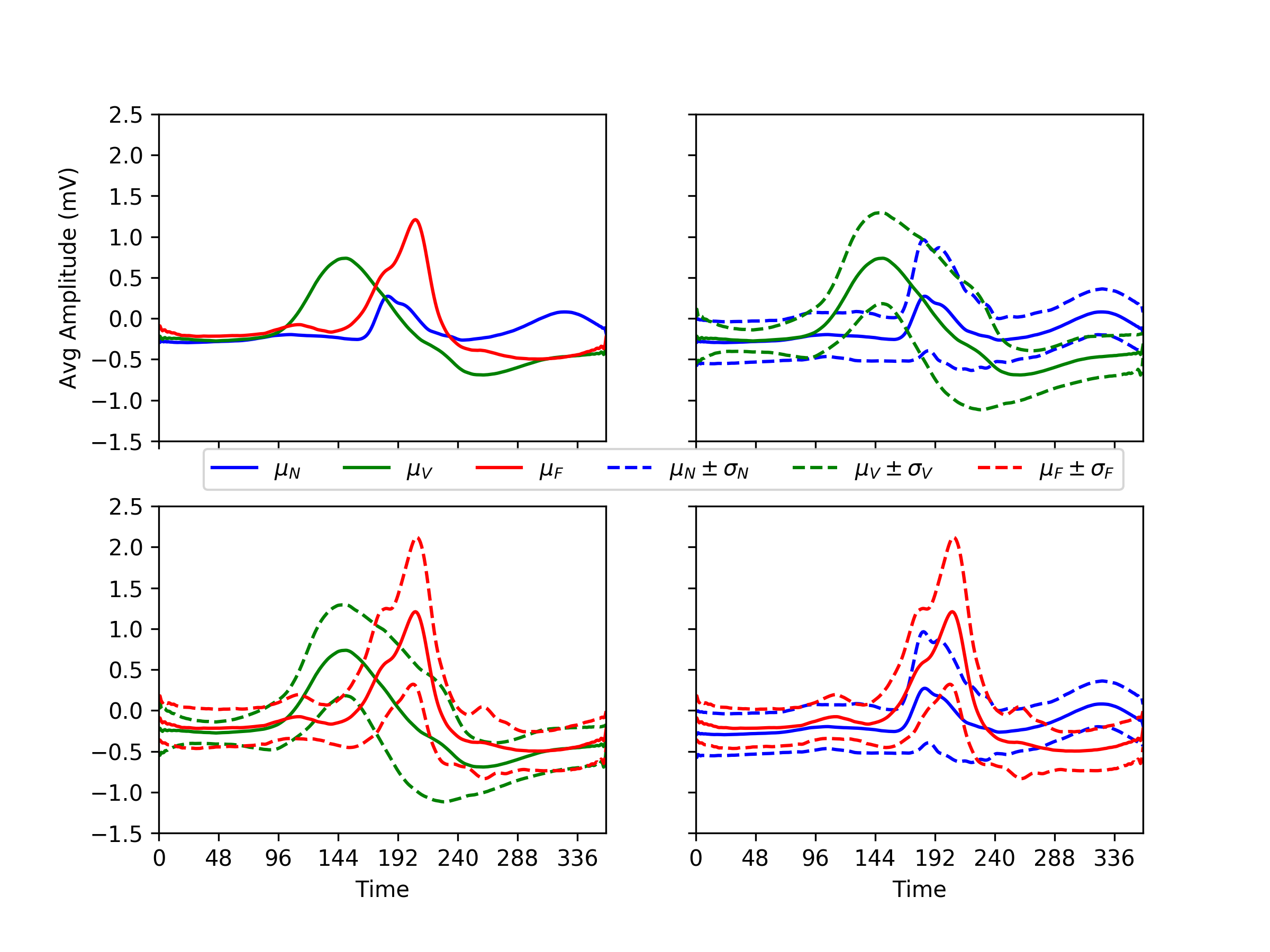}
	\caption{Depiction of pre/post-change parameters for different types of heartbeats for the patient with ID=208.}
	\label{fig:parPrePostChangeBeatsResample}
\end{figure}

\subsubsection{Results of Applying I.P.I.D. Quickest Detection and Isolation Algorithm to ECG Data}
In Fig.~\ref{fig:CusumLikeStartingIndex80} we have plotted the test statistics 
obtained from ECG data with ten heartbeats that began with the heartbeat of the 80th of the test set. Specifically, we plot the statistic in \eqref{eq:repeatwindowlimited} for every class. The red statistic is for arrhythmia of type F and the green statistic is for arrhythmia of type V. A spike in the values of these statistics indicates that an arrhythmia of the corresponding type has been detected. As seen in the figure, the algorithm is quite accurate in detecting arrhythmias. We remark that we reset the test statistic to zero each time the statistic crosses a threshold. 


Next, we consider the segment with starting heartbeat index of 1373 that is represented in Fig. \ref{fig:CusumLikeStartingIndex1373}. As can be seen from the figure, there are both false alarms and incorrect fault isolations. 
Finally, we apply the algorithm to another segment shown in Fig. \ref{fig:CusumLikeStartingIndex235} that begins with a type `V' heartbeat in index 235 and ended with a type `V' in index 244. It only had one miss-classification error, which resulted in isolating the type `F' instead of the type `V'. 
\begin{figure}
	\centering
	\includegraphics[width=0.7\linewidth]{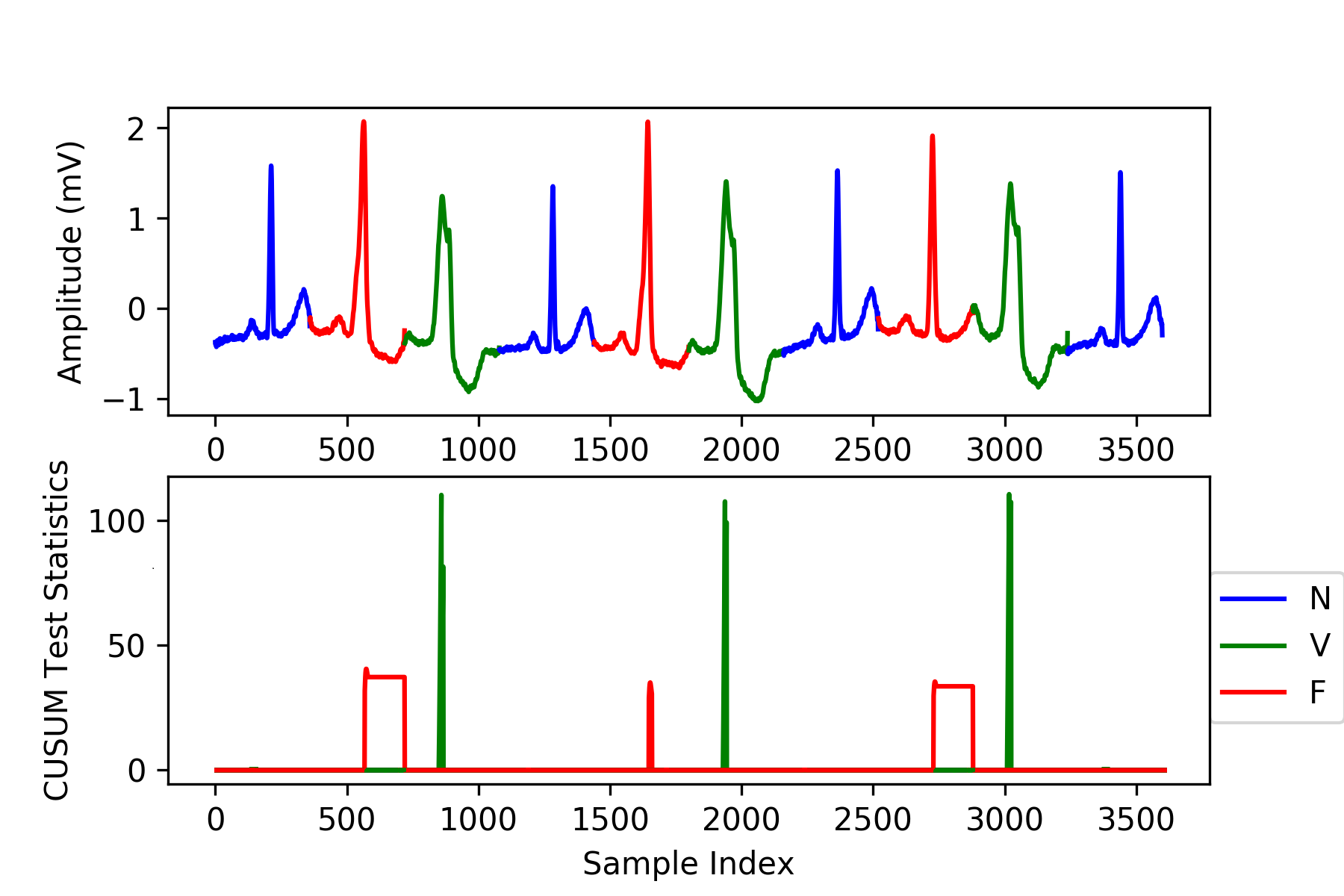}
	\caption{Illustrations of a sample path of ECG with ten heartbeats had started with normal heartbeat with index 80 and calculated i.p.i.d. fault isolation test statistic.}
	\label{fig:CusumLikeStartingIndex80}
\end{figure}

\begin{figure}
	\centering
	\includegraphics[width=0.7\linewidth]{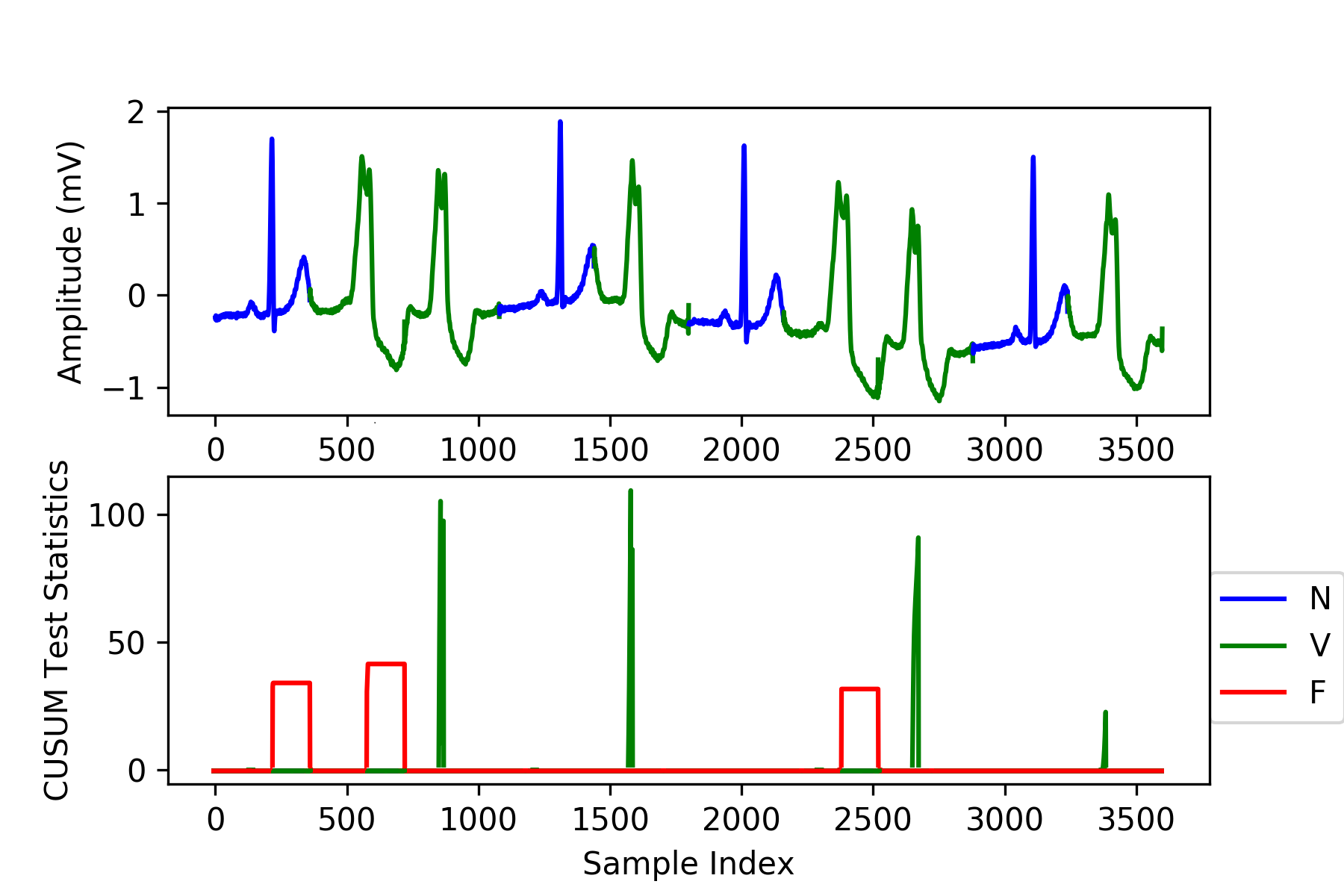}
	\caption{The evolution of i.p.i.d. fault isolation test statistic with one type of arrhythmia, happened at index 1374 and continued with the presence of four arrhythmias identical to type `V' arrhythmia and ended in a type `V' arrhythmia.}
	\label{fig:CusumLikeStartingIndex1373}
\end{figure}

\begin{figure}
	\centering
	\includegraphics[width=0.7\linewidth]{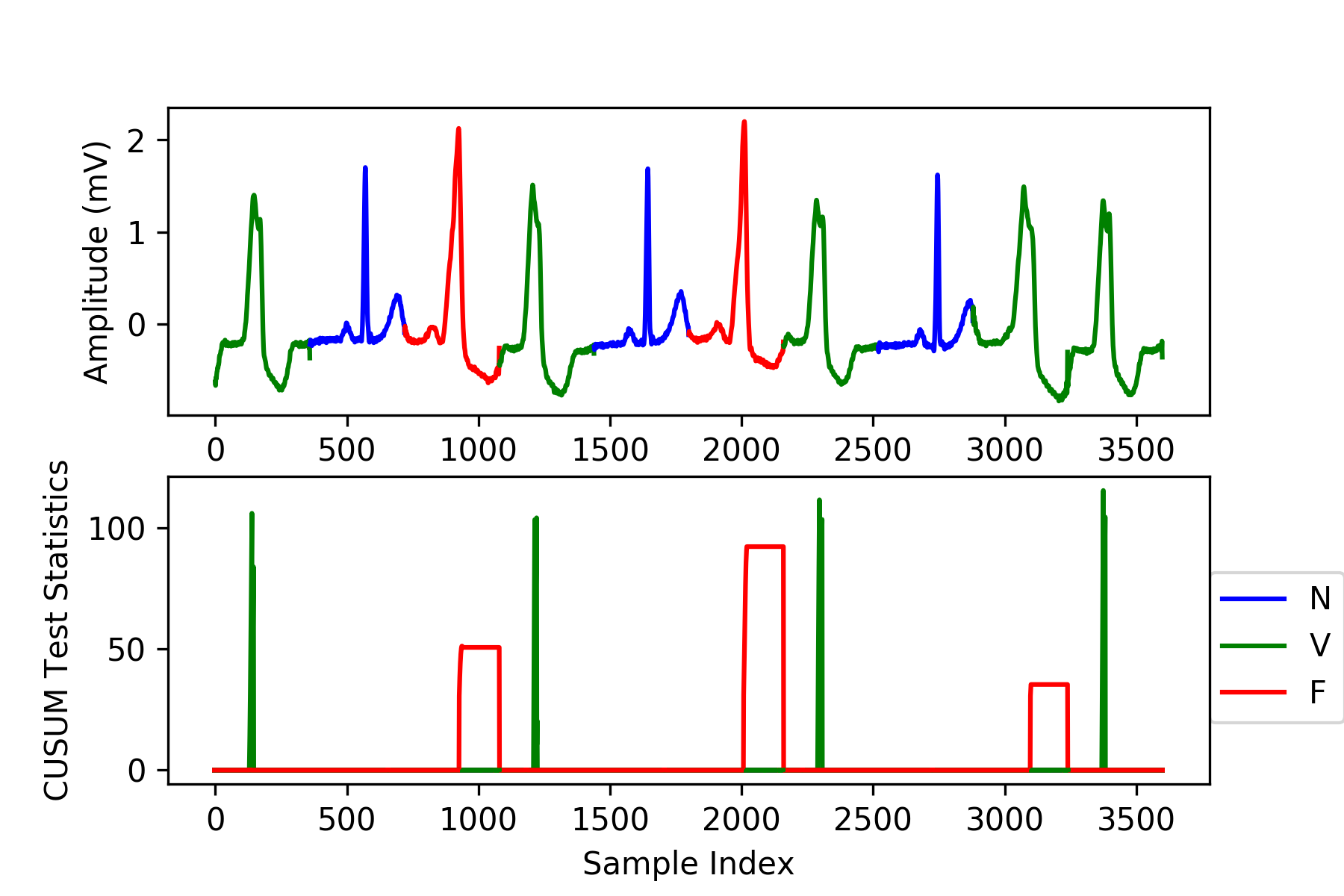}
	\caption{Depiction for ECG segment included ten heartbeats started with arrhythmia at index 235 in the test set.}
	\label{fig:CusumLikeStartingIndex235}
\end{figure}

\subsubsection{Results of Applying I.P.I.D. Quickest Detection and Isolation Algorithm to Wavelet Data}
\label{sec:wavelet}
Due to a limited amount of multi-class data in the MIT-BIH dataset, we use simulated data using wavelets to show the effectiveness of our algorithm for three-class detection and classification. One of the noise-resistant wavelet transformations on ECG was Ricker wavelet or Marr wavelet which is known as Mexican hat or Marr's wavelet in the Americas \cite{GnecchiBioSigProsCtrlJ2017}. For simulation purposes, we used the Mexican-Hat wavelet, which resembles morphological features of ECG heartbeats with known pre and post-change distributions' parameters. In mathematical terms, Marr's wavelet has been formulated as follows.
\[\psi(t)= \frac{2}{\sqrt[4]{9\pi}} \left(1-t^{2}\right) e^{-t^{2} / 2}\]\\
For discrete-time simulations, we re-sampled a 100-long wave centered at zero from a Scipy's Ricker wavelet generating function. Different functional variations of a Mexican hat wavelet, such as shift up, scaling, time delay, and or integration of two perturbations produced three types of anomalies. In total, we had four classes as illustrated in Fig. \ref{fig:MexicanHatFourClassIsolationAnomalies}. The actual data was generated by adding zero-mean Gaussian noise with variance $0.01$ to the wavelets and then cascading the noisy waveforms together to make an ECG-like waveform pattern. The results are plotted in Fig. \ref{fig:CusumLikeSampleMexican} and Fig. \ref{fig:CusumLikeRndSampleMexican}. As seen in the figures, our algorithm can detect and identify faults quite accurately in real time. 
\\
\begin{figure}
	\centering
	\includegraphics[width=0.4\linewidth]{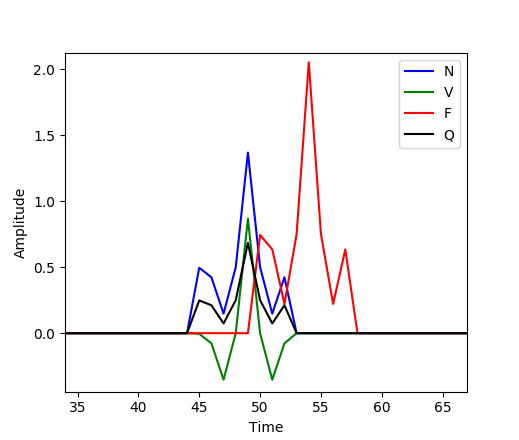}
	\caption{Illustration of a zoomed-into specific time index for four different wavelets by applying a set of transformations on a wavelet.}
	\label{fig:MexicanHatFourClassIsolationAnomalies}
\end{figure}
\\

\begin{figure}
	\centering
	\includegraphics[width=0.7\linewidth]{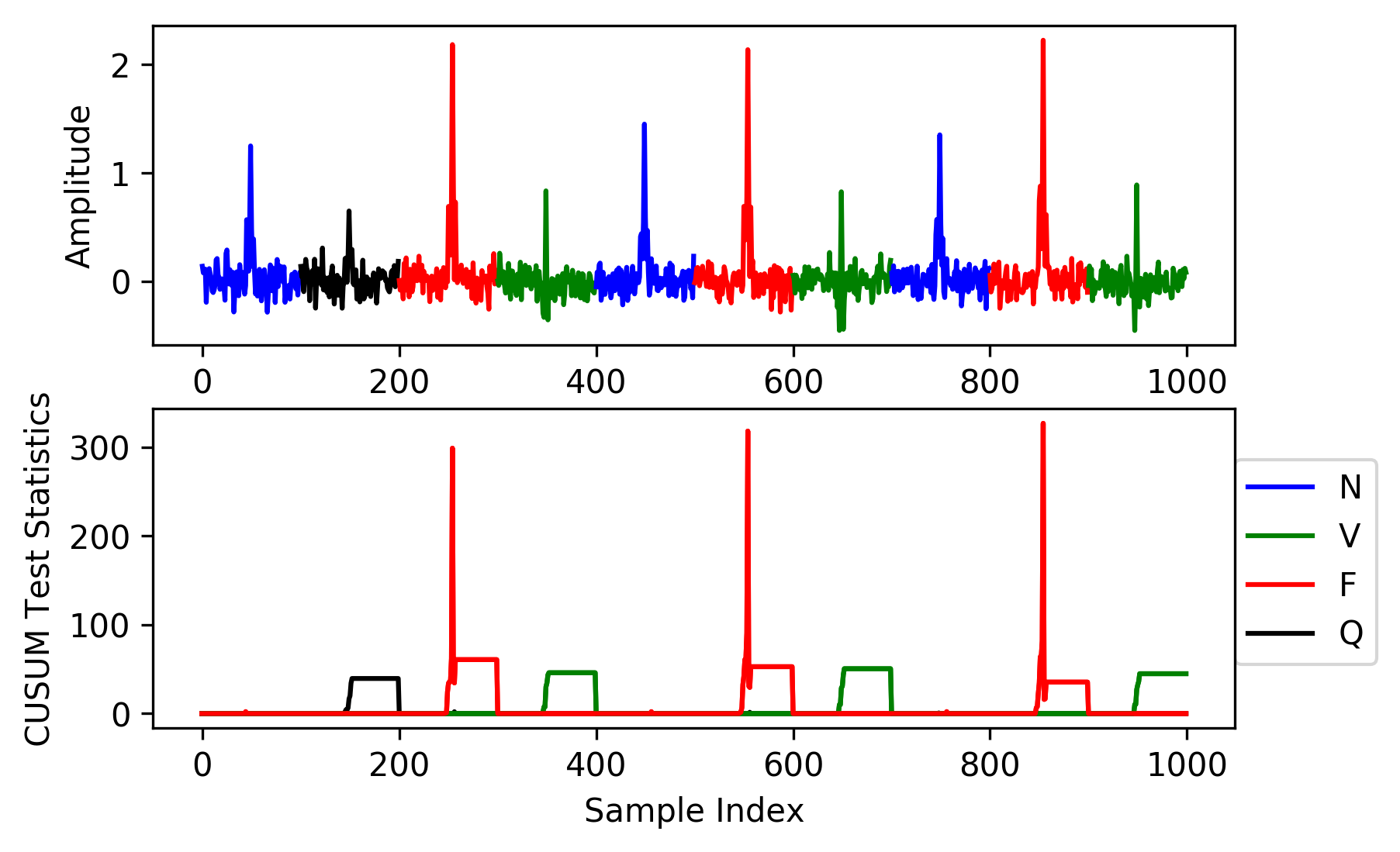}
	\caption{Illustration for a sample path and corresponding i.p.i.d. fault isolation test statistic on a sequence with different types of anomalies.}
	\label{fig:CusumLikeSampleMexican}
\end{figure}

\begin{figure}
	\centering
	\includegraphics[width=0.7\linewidth]{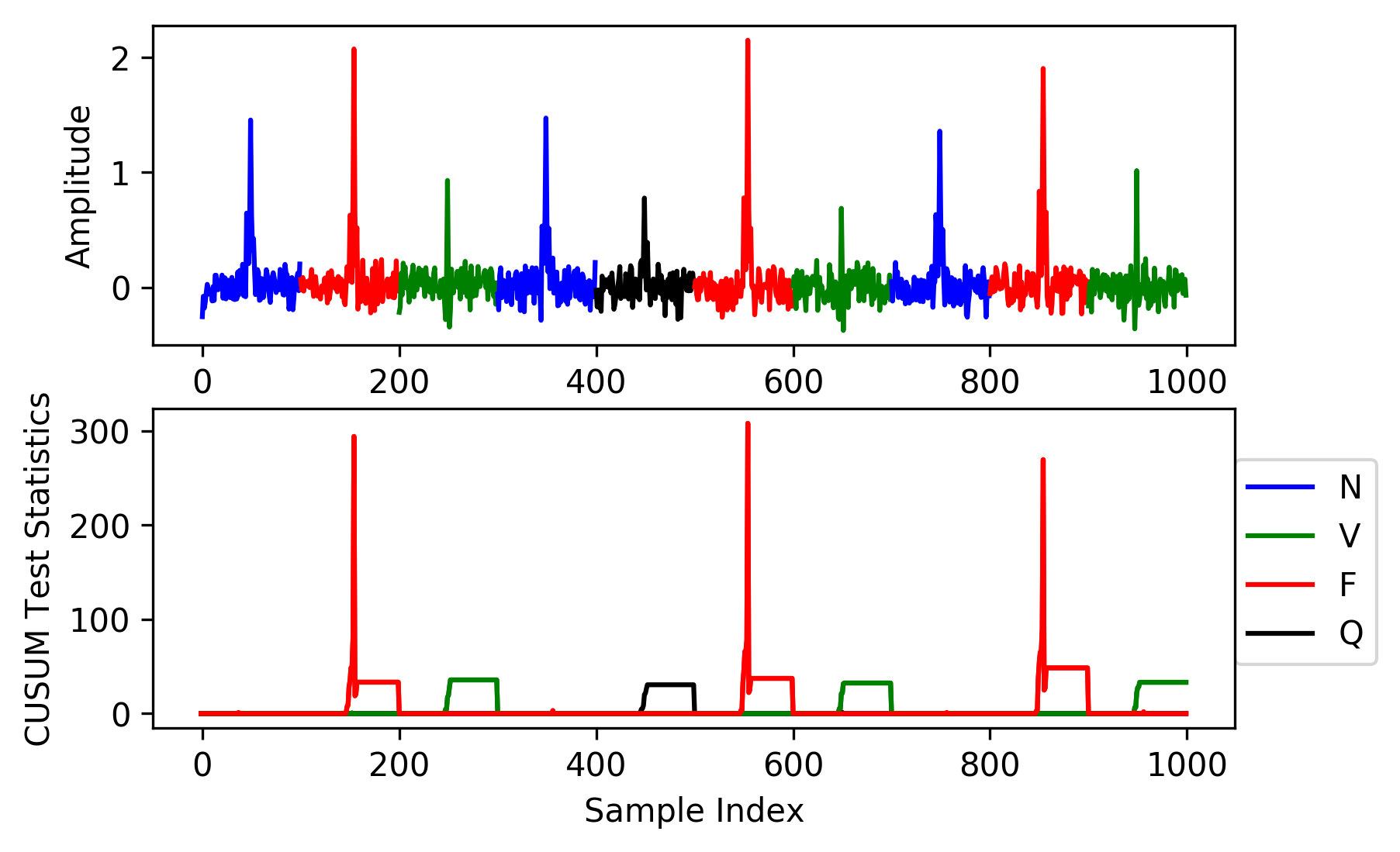}
	\caption{Illustration for different sample paths and corresponding i.p.i.d. fault isolation test statistic.}
	\label{fig:CusumLikeRndSampleMexican}
\end{figure}

\section{Conclusions}
We developed algorithms for the quickest change detection in i.p.i.d. processes when the post-change i.p.i.d. law is unknown. We introduced the concept of a least favorable i.p.i.d. law and showed that a multi-threshold Shiryaev algorithm designed using the least favorable i.p.i.d. law is robust optimal. We then proposed an algorithm for quickest change detection and fault isolation in the i.p.i.d. setting and showed that it is asymptotically optimal, as the rate of false alarms and misclassifications go to zero. We also showed that a mixture-based test is asymptotically optimal for the multislot quickest change detection problem. We showed that the developed algorithm can be successfully used to detect anomalies in real traffic data and real ECG data. 

\section{Acknowledgements}

The work of Yousef Oleyaeimotlagh, Taposh Banerjee and Ahmad
Taha was partially supported by the National Science Foundation under Grant 1917164. The work of Yousef Oleyaeimotlagh, Taposh Banerjee, and Eugene John was also partially
supported by the National Science Foundation under Grant 2041327.

\bibliographystyle{tfs}
\bibliography{TaposhQCD.bib}

\begin{thebibliography}{10}
\providecommand{\MR}{\relax\unskip\space MR }
\providecommand{\url}[1]{\normalfont{#1}}
\providecommand{\urlprefix}{Available at }

\bibitem{bane-NER-2019}
T. Banerjee, S. Allsop, K.M. Tye, D. Ba, and V. Tarokh, \emph{Sequential
  Detection of Regime Changes in Neural Data}, in \emph{Proc. of the 9th
  International IEEE EMBS Conference on Neural Engineering}, Mar. 2019.

\bibitem{bane-isit-2019}
T. Banerjee, P. Gurram, and G. Whipps, \emph{Bayesian quickest detection of
  changes in statistically periodic processes}, in \emph{2019 IEEE
  International Symposium on Information Theory (ISIT)}. IEEE, 2019, pp.
  2204--2208.

\bibitem{bane-icassp-2019}
T. Banerjee, P. Gurram, and G. Whipps, \emph{Quickest detection of deviations
  from periodic statistical behavior}, in \emph{ICASSP 2019-2019 IEEE
  International Conference on Acoustics, Speech and Signal Processing
  (ICASSP)}. IEEE, 2019, pp. 5351--5355.

\bibitem{bane-allerton-2019}
T. Banerjee, P. Gurram, and G. Whipps, \emph{A sequential detection theory for
  statistically periodic random processes}, in \emph{2019 57th Annual Allerton
  Conference on Communication, Control, and Computing (Allerton)}. IEEE, 2019,
  pp. 290--297.

\bibitem{bane-tit-2021}
T. Banerjee, P. Gurram, and G.T. Whipps, \emph{A {Bayesian} theory of change
  detection in statistically periodic random processes}, IEEE Transactions on
  Information Theory 67 (2021), pp. 2562--2580.

\bibitem{bane-tsp-2015}
T. Banerjee and V.V. Veeravalli, \emph{Data-efficient quickest change detection
  in sensor networks}, IEEE Transactions on Signal Processing 63 (2015), pp.
  3727--3735.

\bibitem{bane-globalsip-2018}
T. Banerjee, G. Whipps, P. Gurram, and V. Tarokh, \emph{Cyclostationary
  Statistical Models and Algorithms for Anomaly Detection Using Multi-Modal
  Data}, in \emph{Proc. of the 6th IEEE Global Conference on Signal and
  Information Processing}, Nov. 2018.

\bibitem{bane-fusion-2018}
T. Banerjee, G. Whipps, P. Gurram, and V. Tarokh, \emph{Sequential Event
  Detection Using Multimodal Data in Nonstationary Environments}, in
  \emph{Proc. of the 21st International Conference on Information Fusion}, Jul.
  2018.

\bibitem{bert-dyn-prog-book-2017}
D. Bertsekas, \emph{Dynamic Programming and Optimal Control, Vol. II}, Athena
  Scientific, Belmont, Massachusetts, 2017.

\bibitem{Can2012IEEETransBioEng}
Y. Can, B.V.K.V. Kumar, and M.T. Coimbra, \emph{Heartbeat classification using
  morphological and dynamic features of {E}{C}{G} signals}, IEEE Transactions
  on Biomedical Engineering 59 (2012), pp. 2930--2941.
  \urlprefix\url{https://dx.doi.org/10.1109/tbme.2012.2213253}.

\bibitem{chen-bane-tps-2016}
Y.C. Chen, T. Banerjee, A.D. Domínguez-García, and V.V. Veeravalli,
  \emph{Quickest line outage detection and identification}, IEEE Transactions
  on Power Systems 31 (2016), pp. 749--758.

\bibitem{Dechazal2004IE3BioEng}
P. Dechazal, M. O'Dwyer, and R. Reilly, \emph{Automatic classification of
  heartbeats using {E}{C}{G} morphology and heartbeat interval features}, IEEE
  Transactions on Biomedical Engineering 51 (2004), pp. 1196--1206.

\bibitem{gardner2006cyclostationarity}
W.A. Gardner, A. Napolitano, and L. Paura, \emph{Cyclostationarity: Half a
  century of research}, Signal processing 86 (2006), pp. 639--697.

\bibitem{PhysioNet}
A.L. Goldberger, L.A.N. Amaral, L. Glass, J.M. Hausdorff, P.C. Ivanov, R.G.
  Mark, J.E. Mietus, G.B. Moody, C.K. Peng, and H.E. Stanley,
  \emph{{PhysioBank, PhysioToolkit, and PhysioNet}: Components of a new
  research resource for complex physiologic signals}, Circulation 101 (2000
  (June 13)), pp. e215--e220. Circulation Electronic Pages:
  http://circ.ahajournals.org/content/101/23/e215.full PMID:1085218; doi:
  10.1161/01.CIR.101.23.e215.

\bibitem{GnecchiBioSigProsCtrlJ2017}
J.A. Gutiérrez-Gnecchi, R. Morfin-Magaña, D. Lorias-Espinoza, A. del  Carmen
  Tellez-Anguiano, E. Reyes-Archundia, A. Méndez-Patiño, and R.
  Castañeda-Miranda, \emph{{D}{S}{P}-based arrhythmia classification using
  wavelet transform and probabilistic neural network}, Biomedical Signal
  Processing and Control 32 (2017), pp. 44--56.
  \urlprefix\url{https://www.sciencedirect.com/science/article/pii/S1746809416301677}.

\bibitem{Hannun2019Nature}
A.Y. Hannun, P. Rajpurkar, M. Haghpanahi, G.H. Tison, C. Bourn, M.P. Turakhia,
  and A.Y. Ng, \emph{Cardiologist-level arrhythmia detection and classification
  in ambulatory electrocardiograms using a deep neural network}, Nature
  Medicine 25 (2019), pp. 65--69.
  \urlprefix\url{https://doi.org/10.1038/s41591-018-0268-3}.

\bibitem{lai2000faultiso}
T.L. Lai, \emph{Sequential multiple hypothesis testing and efficient fault
  detection-isolation in stochastic systems}, IEEE Transactions on Information
  Theory 46 (2000), pp. 595--608.

\bibitem{Moody2001}
G.B. Moody and R.G. Mark, \emph{The impact of the {M}{I}{T}-{B}{I}{H}
  arrhythmia database}, IEEE Engineering in Medicine and Biology Magazine 20
  (2001), pp. 45--50.

\bibitem{Moody2005}
G.B. Moody and R.G. Mark, \emph{{M}{I}{T}-{B}{I}{H} arrhythmia database}
  (2005). \urlprefix\url{https://physionet.org/content/mitdb/1.0.0/}, Accessed
  on 18.01.2021.

\bibitem{niki-ieeetit-2003}
I.V. Nikiforov, \emph{A lower bound for the detection/isolation delay in a
  class of sequential tests}, IEEE Trans. Inf. Theory 49 (2003), pp.
  3037--3046.

\bibitem{poor-hadj-qcd-book-2009}
H.V. Poor and O. Hadjiliadis, \emph{Quickest detection}, Cambridge University
  Press, 2009.

\bibitem{shir-opt-stop-book-1978}
A.N. Shiryayev, \emph{Optimal Stopping Rules}, Springer-Verlag, New York, 1978.

\bibitem{taga-jqt-1998}
G. Tagaras, \emph{A survey of recent developments in the design of adaptive
  control charts}, Journal of Quality Technology 30 (1998), pp. 212--231.

\bibitem{tartakovsky2019sequential}
A. Tartakovsky, \emph{Sequential Change Detection and Hypothesis Testing:
  General Non-i.i.d. Stochastic Models and Asymptotically Optimal Rules},
  Chapman and Hall/CRC, 2019.

\bibitem{tart-niki-bass-2014}
A.G. Tartakovsky, I.V. Nikiforov, and M. Basseville, \emph{Sequential Analysis:
  {Hypothesis} Testing and Change-Point Detection}, Statistics, CRC Press,
  2014.

\bibitem{unni-etal-ieeeit-2011}
J. Unnikrishnan, V.V. Veeravalli, and S.P. Meyn, \emph{Minimax robust quickest
  change detection}, IEEE Trans. Inf. Theory 57 (2011), pp. 1604 --1614.

\bibitem{veer-bane-elsevierbook-2013}
V.V. Veeravalli and T. Banerjee, \emph{Quickest Change Detection}, Academic
  Press Library in Signal Processing: Volume 3 -- Array and Statistical Signal
  Processing, 2014.

\bibitem{taha-bane-acc-2020}
S.C. Vishnoi, S.A. Nugroho, A.F. Taha, C. Claudel, and T. Banerjee,
  \emph{Asymmetric cell transmission model-based, ramp-connected robust traffic
  density estimation under bounded disturbances}, in \emph{2020 American
  Control Conference (ACC)}. IEEE, 2020, pp. 1197--1202.

\bibitem{zhang-demba-2018}
Y. Zhang, N. Malem-Shinitski, S.A. Allsop, K. Tye, and D. Ba, \emph{Estimating
  a separably-markov random field (smurf) from binary observations}, Neural
  Computation 30 (2018), pp. 1046--1079.

\end{thebibliography}

\end{document}